\documentclass[11pt]{article}

\usepackage[utf8]{inputenc}
\usepackage[T1]{fontenc}
\usepackage{lmodern}
\usepackage[DIV=11]{typearea} % large value = less margin (roughly speaking), recommended values: 10pt: 8, 11pt: 10, 12pt: 12
\usepackage{microtype}
\usepackage{amssymb}
\usepackage{amsmath}
\usepackage{amsthm}
\usepackage{thmtools}

\usepackage{xcolor}
\usepackage{xspace}
\usepackage{paralist}
\usepackage{comment}

\usepackage[
	style=alphabetic,
	backref=true,
	doi=false,
	url=false,
	maxcitenames=3,
	mincitenames=3,
	maxbibnames=10,
	minbibnames=10,
	backend=bibtex8,
	sortlocale=en_US,
	sorting=nyt,
	bibencoding=ascii
]{biblatex}

\usepackage[ocgcolorlinks]{hyperref} % Option ocgcolorlinks makes links only colored on screen and not on printouts
\newcommand*\samethanks[1][\value{footnote}]{\footnotemark[#1]}

\makeatletter
\g@addto@macro\bfseries{\boldmath}
\makeatother

\makeatletter
\g@addto@macro\mdseries{\unboldmath}
\g@addto@macro\normalfont{\unboldmath}
\g@addto@macro\rmfamily{\unboldmath}
\g@addto@macro\upshape{\unboldmath}
\makeatother

\DeclareCiteCommand{\citem}
    {}
    {\mkbibbrackets{\bibhyperref{\usebibmacro{postnote}}}}
    {\multicitedelim}
    {}

\renewcommand*{\multicitedelim}{\addcomma\space}

\AtEveryCitekey{\clearfield{note}}

\newcommand{\myhref}[1]{%
  \iffieldundef{doi}
    {\iffieldundef{url}
       {#1}
       {\href{\strfield{url}}{#1}}}
    {\href{http://dx.doi.org/\strfield{doi}}{#1}}%
}

\DeclareFieldFormat{title}{\myhref{\mkbibemph{#1}}}
\DeclareFieldFormat
  [article,inbook,incollection,inproceedings,patent,thesis,unpublished]
  {title}{\myhref{\mkbibquote{#1\isdot}}}

\makeatletter
\AtBeginDocument{%
    \newlength{\temp@x}%
    \newlength{\temp@y}%
    \newlength{\temp@w}%
    \newlength{\temp@h}%
    \def\my@coords#1#2#3#4{%
      \setlength{\temp@x}{#1}%
      \setlength{\temp@y}{#2}%
      \setlength{\temp@w}{#3}%
      \setlength{\temp@h}{#4}%
      \adjustlengths{}%
      \my@pdfliteral{\strip@pt\temp@x\space\strip@pt\temp@y\space\strip@pt\temp@w\space\strip@pt\temp@h\space re}}%
    \ifpdf
      \typeout{In PDF mode}%
      \def\my@pdfliteral#1{\pdfliteral page{#1}}% I don't know why % this command...
      \def\adjustlengths{}%
    \fi
    \ifxetex
      \def\my@pdfliteral #1{}% isn't equivalent to this one
      \def\adjustlengths{\setlength{\temp@h}{-\temp@h}\addtolength{\temp@y}{1in}\addtolength{\temp@x}{-1in}}%
    \fi%
    \def\Hy@colorlink#1{%
      \begingroup
        \ifHy@ocgcolorlinks
          \def\Hy@ocgcolor{#1}%
          \my@pdfliteral{q}%
          \my@pdfliteral{7 Tr}% Set text mode to clipping-only
        \else
          \HyColor@UseColor#1%
        \fi
    }%
    \def\Hy@endcolorlink{%
      \ifHy@ocgcolorlinks%
        \my@pdfliteral{/OC/OCPrint BDC}%
        \my@coords{0pt}{0pt}{\pdfpagewidth}{\pdfpageheight}%
        \my@pdfliteral{F}% Fill clipping path (the url's text) with
        \my@pdfliteral{EMC/OC/OCView BDC}%
        \begingroup%
          \expandafter\HyColor@UseColor\Hy@ocgcolor%
          \my@coords{0pt}{0pt}{\pdfpagewidth}{\pdfpageheight}%
          \my@pdfliteral{F}% Fill clipping path (the url's text)
        \endgroup%
        \my@pdfliteral{EMC}%
        \my@pdfliteral{0 Tr}% Reset text to normal mode
        \my@pdfliteral{Q}%
      \fi
      \endgroup
    }%
}
\makeatother

\colorlet{DarkRed}{red!50!black}
\colorlet{DarkGreen}{green!50!black}
\colorlet{DarkBlue}{blue!50!black}

\hypersetup{
	linkcolor = DarkRed,
	citecolor = DarkGreen,
	urlcolor = DarkBlue,
	bookmarksnumbered = true,
	linktocpage = true
}

\declaretheorem[numberwithin=section]{theorem}
\declaretheorem[numberlike=theorem]{lemma}

\declaretheorem[numberlike=theorem]{corollary}
\declaretheorem[numberlike=theorem]{definition}
\declaretheorem[numberlike=theorem]{observation}

\newcommand{\R}{\mathbb{R}}
\def\B{\mathcal{B}}
\def\M{\mathcal{M}}
\newcommand{\set}[1]{\{#1\}}

\renewcommand{\epsilon}{\varepsilon}

\newcommand{\gs}{\textsc{gs}}
\newcommand{\we}{\textsc{we}}
\newcommand{\newwe}{\textsc{reserve-we}\textup{(}r\textup{)}}
\newcommand{\lp}{\textsc{welfare-lp}}
\newcommand{\newlp}{\textsc{reserve-lp}\textup{(}r\textup{)}}

\newcommand{\m}{\textsc{ef}}
\newcommand{\OR}{\textsc{or}}

\newcommand{\lu}{\textup{(}}
\newcommand{\ru}{\textup{)}\xspace}
\newcommand{\upbr}[1]{\lu #1\ru}

\DeclareMathOperator*{\argmax}{arg\,max}
\DeclareMathOperator*{\argmin}{arg\,min}
\newcommand{\ignore}[1]{}

\newif\iffullversion
\newcommand{\infull}[1]{\iffullversion #1\fi}
\newcommand{\inshort}[1]{\iffullversion \else #1\fi}
\fullversionfalse
\fullversiontrue

\usepackage[retainorgcmds]{IEEEtrantools}
\IEEEeqnarraydefcolsep{0}{3cm}
\makeatletter
\newcommand\ztag[1]{%
\def\@currentlabel{#1}%
\gdef\tmp{%
\addtocounter{equation}{-1}%
\def\theequation{#1}}%
\aftergroup\aftergroup\aftergroup\aftergroup\aftergroup\aftergroup
\aftergroup\aftergroup\aftergroup\aftergroup\aftergroup\aftergroup
\aftergroup\aftergroup\aftergroup\aftergroup\aftergroup\aftergroup
\aftergroup\aftergroup\aftergroup\aftergroup\aftergroup\aftergroup
\aftergroup\aftergroup\aftergroup\aftergroup\aftergroup\aftergroup
\aftergroup
\tmp}

\usepackage{booktabs}

\title{Ad Exchange: \\Envy-Free Auctions with Mediators\thanks{
A short version of this paper appeared in \cite{BenzwiHL15}.}}

\author{
Oren Ben-Zwi\thanks{emarsys labs, emarsys eMarketing Systems AG, Vienna, Austria.
Work done while at the University of Vienna, Faculty of Computer Science.}
\and
Monika Henzinger\thanks{University of Vienna, Faculty of Computer Science, Vienna, Austria.}
\and
Veronika Loitzenbauer\samethanks[3]
}

\date{}

\hypersetup{
	pdftitle = {Ad Exchange: Envy-Free Auctions with Mediators},
	pdfauthor = {Oren Ben-Zwi, Monika Henzinger, Veronika Loitzenbauer}
}

\begin{document}
\maketitle
\begin{abstract}
Ad exchanges are an emerging platform for trading advertisement slots on the web 
with billions of dollars revenue per year. Every time  a user visits a web page, the publisher
of that web page can ask an ad exchange to auction off the ad slots on this page
to determine which advertisements 
are shown at which price. Due to the high volume of traffic, ad networks typically act as mediators for 
individual advertisers at ad exchanges. If multiple advertisers in an ad network are interested
in the ad slots of the same auction, the ad network might use a ``local'' auction 
to resell the obtained ad slots among its advertisers.

In this work we want to deepen the theoretical understanding of these new markets
by analyzing them from the viewpoint of combinatorial auctions.
Prior work studied mostly single-item auctions, while we allow
the advertisers to express richer preferences over multiple items.
We develop a game-theoretic model for the entanglement of the {\em central} auction
at the ad exchange with the {\em local} auctions at the ad networks. We consider the
incentives of all three involved parties and suggest a \emph{three-party competitive equilibrium},
an extension of the Walrasian equilibrium that ensures envy-freeness for all participants.
We show the existence of a three-party competitive equilibrium and a polynomial-time algorithm
to find one for gross-substitute bidder valuations.
\end{abstract}

% \begin{keyword}
% 	advertising exchange \sep combinatorial auctions \sep gross substitute valuations
% 	\sep Walrasian equilibrium \sep three-party equilibrium \sep 
% 	auctions with intermediaries
% \end{keyword}

%
\section{Introduction}
As advertising on the  web becomes more  mature, 
\emph{ad exchanges} (AdX) play a growing role as a platform for selling 
advertisement slots from publishers to advertisers~\cite{KorulaMN15}. 
Following the Yahoo! acquisition 
of Right Media in 2007, all major web companies, such as Google, Facebook, and 
Amazon, have created or acquired their own ad exchanges. 
%Ebay, Microsoft, and Twitter
Other major ad exchanges are provided by the Rubicon
Project, OpenX, and AppNexus. 
% In 2012 the total revenue
% at ad exchanges was estimated to be around two billion dollars~\cite{Rev}.
Every time a  user visits a web page, the publisher
of that web page can ask an ad exchange to auction off the ad slots on this page. 
Thus, the goods traded at an ad exchange are {\em ad impressions}.
This process is also known as \emph{real-time bidding} (RTB).
A web page might contain multiple ad slots, which are currently modeled to be 
sold separately in individual auctions.
Individual advertisers typically do not directly participate in these auctions 
but entrust some ad network to bid on their behalf.
When a publisher sends  an ad impression to an exchange, 
the exchange usually contacts several ad networks and runs a (variant of a) second-price auction~\cite{MansourMN2012} between them,
potentially with a reserve price under which the impression is not sold.
An ad network (e.g. Google's Display Network~\cite{GDN}) might then run 
a second, ``local'' auction to determine the allocation of the
ad slot among its advertisers. We study this 
interaction of a \emph{central auction} at the exchange and \emph{local 
auctions} at the ad networks.\footnote{In this work an auction is an 
algorithm to determine prices of items and their allocation to bidders.}

%As there are only milliseconds between a user request for a web page and the time
%the advertisements are displayed, many engineering questions
%arise (see e.g.~\cite{Muthukrishnan2009,LangEtAl2011,ChenBAD2011}), which we will ignore. Instead we 
% We focus on the game-theoretic analysis of ad exchanges.
% Different  models have been suggested in the last 
% years, see our detailed discussion at the end of this section.
%~\cite{Muthukrishnan2009,FeldmanMMP2010,BalseiroFMM2011,MansourMN2012,BalseiroBW15}.
We develop a game-theoretic model that considers the 
incentives of the following three parties: (1) the ad exchange, (2)  the ad networks, 
and (3) the advertisers. 
As the ad exchange usually charges a fixed percentage of the revenue and hands
the rest to the publishers, the ad exchange and the publishers have the same objective
and can be modeled as one entity.
We then study equilibrium concepts of this new model of a three-party exchange. 
Our model is  described as an ad exchange, but it may also model other scenarios 
with mediators that act between bidders and sellers, as noted already by 
Feldman~et~al.~\cite{FeldmanMMP2010}.
The main differences between our model and earlier models (discussed in detail 
at the end of this section) are the following:
(a) We consider the incentives of all three parties {\em simultaneously}. 
(b) While most approaches in prior work use Bayesian assumptions, we apply {\em  worst-case analysis}.
(c) We allow auctions with {\em multiple heterogeneous items}, namely combinatorial 
auctions, in contrast to the single-item auctions studied so far. 
Multiple items arise naturally when selling ad slots on a per-impression basis, 
since there are usually multiple advertisement slots on a web page.
% As noted already by~\cite{Muthukrishnan2009}, the values of the
% advertisers for different positions and sizes of ad slots on a web page, 
% and combinations thereof, might be complex and values might vary widely among advertisers.
%A more detailed discussion of related work follows at the end of this section.

To motivate the incentives of ad networks and exchanges, we
compare next their short and long-term revenue considerations, following Mansour~et~al.~\cite{MansourMN2012} and Muthukrishnan~\cite{Muthukrishnan2009}.
% Like a stock exchange, an ad exchange brings together buyers and sellers
% in an efficient way. For example, an exchange enables a small company selling
% a niche product to run a world-wide advertising campaign specifically targeted
% at certain potential costumers. On the other hand,
% the publishers can show ads of exactly these advertisers that value their users
% most, which pleases their users and increases their revenue at the same time.
% Due to this increased efficiency,
% % in the matching of user impressions with advertisements, 
% higher prices per ad slot are expected at the exchange than through alternative channels.
Ad exchanges and ad networks generate revenue as follows:
(1) An ad exchange usually receives some percentage of the price paid by the winner(s) of 
the central auction. 
(2) An ad network can charge a higher price to 
its advertisers than it paid to the exchange or it can be paid via direct
contracts with its advertisers.
Thus both the ad exchange and the ad networks (might) profit from higher prices in their auctions.
However, they also have a motivation not to charge too high 
prices as (a) the advertisers could stick to alternative advertising channels such
as long-term contracts with publishers, and 
%, i.e., they are not required to participate in the exchange at all. 
(b) there is a significant competition between 
the various ad exchanges and ad networks, as advertisers can easily switch to 
a competitor.
% in case they are not satisfied with the prices or the allocations. 
Thus, lower prices (might) increase advertiser participation and,
% in general and for an ad exchange and ad network and ,
hence, the long-term revenue of ad exchanges and ad networks.
% add lower risk motivation for ad networks (as discussed on 22.11.2013)?
We only consider a single auction (of multiple items)
% at the exchange. We 
and leave it as an open question to study changes over time.
% in this three-party market, such as the effect of advertisers switching to competing ad networks. 
We still take the long-term considerations outlined above into account by
assuming that the ad exchange aligns its strategic behavior with its long-term
revenue considerations and only desires for each central auction
to sell {\em all} items.\footnote{Our model and results can be adapted to include
reserve prices under which the ad exchange is not willing to sell an item.}
In our model the incentive of an ad network to participate in the exchange
comes from the opportunity to purchase some items at a low price and then
resell them at a higher price. However, due to long-term considerations, 
our model additionally requires
the ad networks to ``satisfy their advertisers'' by faithfully representing
the advertisers' preferences towards the exchange, while still allowing the 
ad networks to extract revenue from the competition between the advertisers 
in their network.\footnote{We implicitly assume that the central auction prices are
accessible to the advertisers such that they can verify whether an ad
network represented their preferences correctly. Informally, we suggest that if 
one ad network ``satisfies its advertisers'' then, over time, all ad networks 
have to follow this behavior to keep their advertisers.}
An example for this kind of restriction for an ad network is 
Google's Display Network~\cite{GDN} that guarantees its advertisers that
each ad impression is sold via a second-price auction, independent of whether
an ad exchange is involved in the transaction or not~\cite{MansourMN2012}.

To model a \emph{stable} outcome in a three-party exchange,
we use the equilibrium concept of
\emph{envy-freeness} for all three types of participants. A participant is
envy-free if he receives his most preferred set of items under the current prices. 
Envy-freeness for all participants is a natural notion to express stability in a market,
as it implies that no coalition of participants would strictly profit from deviating
from the current allocation and prices (assuming truthfully reported preferences).
Thus an envy-free equilibrium supports stability in the market prices, which
in turn facilitates, for example, revenue prediction for prospective participants
and hence might increase participation and long-term revenue. 
For only two parties,
i.e., sellers and buyers, where the sellers have no intrinsic value for the
items they sell, envy-freeness for all participants is equal to a \emph{competitive}
or \emph{Walrasian} equilibrium~\cite{Walras1874}, a well established notion in economics
to characterize an equilibrium in a market where demand equals supply.
We provide a generalization of this equilibrium concept to three parties.

\paragraph*{Our Contribution}
%\paragraph*{Model {\upshape (}informal{\upshape )}.}
We introduce the following model for ad exchanges.
A {\em central seller} wants to sell $k$ \emph{items}. There are $m$ \emph{mediators}
$\M_i$, each with her  own $n_i$ \emph{bidders}. Each bidder has a valuation function over all subsets of the items. 
In the ad exchange setting, the central seller is the ad exchange, the items are 
the ad slots shown to a visitor of a web page, the mediators are the ad networks, 
and the bidders are the advertisers. A bidder does not have any direct 
``connection'' to the central seller. Instead, all  communication is done 
through the mediators. A mechanism 
for allocating the items to the bidders is composed of a \emph{central
auction} with mediators acting as bidders, and then
\emph{local auctions}, one per mediator, in which every mediator allocates the set of items she
bought in the central auction; that is, an auction where the bidders of that mediator are 
the only participating bidders and the items that the mediator received in the 
central auction are the sole items.
The prices of the items obtained in the central auction 
% are assumed to be public information and 
provide a lower bound for the prices in the local auctions, i.e., they act as 
reserve prices in the local auctions.
We assume that the central seller and the bidders have quasi-linear utilities, 
i.e., utility 
functions that are linear in the price, and that their incentive is to maximize
their utility. For the central seller 
this means that his utility from selling a set of slots is just the sum of 
prices of the items in the set.
%, i.e., he has a value of zero for any subset of items, hence his utility is quasi-linear in his value plus the price.
The utility of a bidder on receiving a set of items~$S$ is his value for~$S$ 
minus the sum of the prices of the items in~$S$.

The incentive of a mediator, however, is not so 
straightforward and needs to be defined carefully.
%Note first that a bidder may choose any mediator he wishes.
%only after this selection, he signs some kind of an ``obligating contract'' with the mediator. 
In our model,  to ``satisfy'' her  bidders, each mediator guarantees her  bidders that the outcome of the local auction
will be \emph{minimal envy free}, that is, for the final local price
vector, the item set that is allocated to any bidder is one of his most desirable sets over \emph{all} possible item sets
(even sets that contain items that were not allocated to his mediator, i.e., each 
bidder is not only {\em locally}, but {\em globally envy-free})
and there is no (item-wise) smaller price vector that fulfills this requirement. 
We assume that each mediator wants to maximize her revenue\footnote{For the 
purpose of this paper, the terms revenue and utility are interchangeable.} and 
define the revenue of a mediator for a set of items~$S$ as the difference 
between her earnings when selling~$S$ with this restriction and the price she 
has to pay for~$S$ at the central auction.

For this model we define a new equilibrium concept, namely the 
\emph{three-party competitive equilibrium}.
At this equilibrium all three types of participants are envy-free. Envy-free 
solutions for the bidders always exist, as one can set the 
prices of all items high enough so that no bidder will demand any item. 
Additionally, we require that there is no envy for the central seller, 
meaning that all items are sold.  If there were no mediators, then a two-party envy-free
solution would be exactly a \emph{Walrasian equilibrium}, which for certain 
scenarios can be guaranteed~\cite{KelsoCr1982}.
However, with mediators it is not a-priori 
clear that a three-party competitive equilibrium exists as,
additionally, the mediators have to be envy-free.
We show that for our definition of a mediator's revenue (a) the above requirements
are fulfilled and (b) a three-party competitive equilibrium 
exists whenever a Walrasian equilibrium for the central auction exists
or whenever a two-party equilibrium exists for the bidders and the central seller 
without mediators.
%%OREN: introduce GS - as famous/celebrated/important/etc.
Interestingly, we show that for gross-substitute bidder valuations the incentives 
of this kind of mediator can be represented with an {\OR}-valuation over 
the valuations of her  bidders.
This then leads to the following result:
{\em For {gross-substitute} bidder valuations a three-party competitive equilibrium 
can be computed in polynomial time.} In particular, we will show how to compute the 
three-party competitive equilibrium with minimum prices.

% As a side result we show the existence of a Walrasian equilibrium for 
% gross-substitute valuations without the widespread monotonicity assumption.
% We further show that valuations with reserve prices, which are needed in the local auctions of the mediators,
% can be interpreted as certain non-monotone valuation functions. We prove that
% many of Gul and Stacchetti's results~\cite{GulSt2000}, 
% such as the lattice structure of Walrasian prices, which they proved
% for monotone gross-substitute valuations, also hold for these non-monotone valuations, 
% and, thus, for valuations with reserve prices. Note that non-monotone valuations 
% are interesting on their own as there are economies in which the disposal of 
% unwanted items is not free due to storage expenses or administrative costs.
%add somewhere that the aggregated price vector is also minimal for GS? (with proof)

\paragraph*{Related Work}
%We focus below on related work concerning ad exchanges. An introductory survey to 
%combinatorial auctions can be found in~\cite{AGT}. 
The theoretical research on ad exchanges was initialized by a survey 
of Muthukrishnan~\cite{Muthukrishnan2009} that lists several interesting research directions.
Our approach specifically addresses his 9th problem, namely to enable the 
advertisers to express more complex preferences that arise when multiple
advertisement slots are auctioned off at once as well as to design suitable 
auctions for the exchange and the ad networks to determine allocation
and prices given these preferences. 

The most closely related work with respect to the model of the ad exchange 
is Feldman~et~al.~\cite{FeldmanMMP2010}. It is similar to our work in two aspects: (1) The mediator bids on behalf 
of her  bidders in a central auction and the demand of the mediator as well as 
the tentative allocation and prices for reselling to her  bidders are 
determined via a local auction. (2)
The revenue of the mediator is the price she can obtain from reselling minus 
the price she paid in the central auction. The main differences are:
(a) Only one item is auctioned 
at a time and thus the mediator can determine her  valuation with a single 
local auction.
(b) Their work does not consider the incentives of the bidders, only of the mediators and the central seller.
(c) A Bayesian setting is used where the mediators and 
the exchange know the probability distributions of the bidders' valuations.
Based on this information, the mediators and the exchange choose reserve 
prices for their second-price auctions to maximize their revenue. The work
characterizes the equilibrium strategies for the selection of the reserve 
prices. 

Mansour~et~al.~\cite{MansourMN2012} (mainly) describe the auction at the DoubleClick
exchange. Similar to our work, advertisers use ad networks as mediators for 
the central auction. They observe that if mediators that participate in a 
single-item, second-price 
central auction are only allowed to submit a single bid, then it is not possible 
for the central auction to correctly implement a second-price auction over 
{\em all} bidders as the bidders with the highest and the second highest value 
might use the same mediator. Thus they introduce the
\emph{Optional Second Price} auction, where  every mediator is allowed
to optionally submit the second highest bid with the highest bid. In such an 
auction each mediator can guarantee to her bidders that if one of them is 
allocated the item, then he pays the (global) second-price for it. For the 
single-item setting, the bidders in their auction and in our auction pay the 
same price. If the mediator of the winning bidder did \emph{not} specify an optional 
second price, then her revenue will equal the revenue of our mediator. If she 
did, her revenue will be zero and the central seller will receive the gain 
between the prices in the local and the central auction.

%Thus, an  \emph{Optional Second Price} auction
%As the ad networks influence the information 
%available to the exchange, a true second-price auction is not possible. 
%However, some ad networks~\cite{GDN}
%obligate themselves to determine the assignments of advertisement slots by a
%second-price auction. The \emph{Optional Second Price} auction proposed
%by \cite{MansourMN2012} enables this by allowing the networks to specify an
%optional second bid and charging the winner the maximum of the bids of the 
%other networks and her own second bid. They do not prove any non-obvious
%properties of their auction.

Stavrogiannis~et~al.~\cite{StavrogiannisGP2013} consider a game between bidders and mediators, 
where the bidders can select mediators (based on Bayesian assumptions of 
each other's valuations) and the mediators can set the reserve prices in the 
second-price local auction. The work presents mixed Nash equilibrium strategies 
for the bidders to select their mediator. In~\cite{StavrogiannisGP2014} the same
authors compare different single-item local auctions with respect to the achieved
social welfare and the revenue of the mediators and the exchange.
 
Balseiro~et~al.~(2015) introduced a setting that 
% includes the incentives of bidders, but their model 
does {\em not} include mediators~\cite{BalseiroBW15}.
Instead, they see the ad exchange as a game between publishers, who select 
parameters such as reserve prices for second-price auctions, and advertisers,
whose budget constraints link different auctions over time. They introduced
a new equilibrium concept for this game and used this to analyze the impact
of auction design questions such as the selection of a reserve price.
%
%and~\cite{GhoshMPV2009}
Balseiro~et~al.~(2014)~\cite{BalseiroFMM14} and Dvo{\v r}{\' a}k and Henzinger~\cite{DvorakH14} 
studied a publisher's trade-off between using an ad exchange versus fulfilling long-term contracts with advertisers. 

% The properties of ascending auction in general have been studied in~\cite{Milgrom2000}
% and for gross-substitute valuations in particular in~\cite{GulSt2000,KelsoCr1982,PaesLeme2014,Ben-ZwiLaNe2013}.
% Mediators as entrusted parties, in contrast to players with incentives, appear 
% for example in \cite{MondererT2009} and~\cite{AshlagiMT2009}.
% Cavallo McAfee Vassilvitskii - Ad Auction Workshop EC '12 (no proceedings):
% mediators charge per click and pay publisher per view and gain arbitrage

% \cite{ChakrabortyEGM2010} develop a policy for the exchange when to contact
% an ad network to participate in the central auction when each network can 
% only handle a certain amount of auctions per time period.

Equilibria in trading networks (such as ad exchanges) are also 
addressed in the ``matching with contracts'' literature.
Hatfield and Milgrom~\cite{HatfieldMi2005} presented a new model where instead of 
bidders and items there are \emph{agents} and \emph{trades} between pairs of 
agents. The potential trades are modeled as edges in a graph where the
agents are represented by the nodes. 
Agent valuations are then defined over the potential trades and assumed to be 
monotone substitute. They proved the existence of an (envy-free) equilibrium when
the agent-trades graph is bipartite. Later this was improved to directed
acyclic graphs by Ostrovsky~\cite{Ostrovsky2008} and to arbitrary 
graphs by Hatfield~et~al.~\cite{HatfieldKoNiOsWe2013}.
They did not show (polynomial-time) algorithms to reach equilibria.
Our model can be reduced to this model, hence a three-party equilibrium exists when 
all bidders are monotone gross substitute. 
% The result of this reduction (not stated here) is not polynomial in the number of bidders and items. 
However, we are not aware of a reduction that is polynomial in the
number of bidders and items.

\paragraph*{Outline}
%% suggestion Oren:
% In section 3 we present our new model and a thorough discussion of its merits.
% In section 4 we present our main equivalence and algorithmic results.
% In section 5 we conclude and suggest future directions.
We formally define
our model for ad exchanges in Section~\ref{sec:model}. 
% In particular we define
% the three-party competitive equilibrium and the mediator model of our choice.
% We give two simple existence proofs for the three-party equilibrium; one under the 
% condition that a two-party equilibrium between the bidders and the central seller 
% exists, and one for the case that a two-party equilibrium between the central 
% seller and the mediators exists.
In Section~\ref{sec:alg} we present our main results for gross-substitute bidders,
% focus on gross-substitute bidders and show that
% in this case our mediator definition is equivalent to a valuation-based
% $\OR$-player definition and how this implies 
including a polynomial-time algorithm to 
compute a three-party competitive equilibrium.
Finally we conclude and suggest future directions in Section~\ref{sec:discuss}.

\section{Preliminaries}\label{sec:prelim}
Let $\Omega$ denote a set of $k$ items.
A \emph{price vector} is an assignment of a non-negative price to every element
of $\Omega$. For a price vector $p = (p_1, ..., p_k)$ and a set $S \subseteq
\Omega$ we use $p(S) = \sum_{j \in S} p_j$. 
For any two price vectors $p, r$ an inequality such as $p \ge r$ as well
as the operations $\min(p,r)$ and $\max(p,r)$ are meant item-wise.

We denote with $\langle \Omega_b\rangle = \langle \Omega_b\rangle_{b\in \B}$ an \emph{allocation} of the items in 
$\Omega$ such that for all bidders $b \in \B$ the set of
items allocated to $b$ is given by $\Omega_b$ and we have 
$\Omega_b \subseteq \Omega$ 
and $\Omega_b \cap \Omega_{b'} = \emptyset$ for $b' \ne b$, $b' \in \B$.
Note that some items might not be allocated to any bidder.

A \emph{valuation} function $v_b$ of a bidder~$b$ is a function from 
$2^{\Omega}$ to~$\R$, where $2^{\Omega}$ denotes the set of all subsets 
of~$\Omega$. We assume throughout the paper $v_b(\emptyset) = 0$.
Unless specified otherwise, for this work we assume \emph{monotone} valuations, that is, for $S\subseteq T$ we have $v_b(S)\leq v_b(T)$.
This assumption is made for ease of presentation.
% and it is not necessary as we show in the full version of this work.
We use $\{v_b\}$ to denote a collection of valuation functions.
The \inshort{(quasi-linear) }\emph{utility} of a bidder~$b$ from a set $S \subseteq \Omega$ at prices $p \ge 0$ is defined 
as $u_{b,p}(S) = v_b(S) - p(S)$. \infull{Such utility functions are often called 
\emph{quasi-linear}, i.e., linear in the price.}
The \emph{demand} $D_b(p)$ of a bidder $b$ for prices $p \ge 0$ is the set of subsets 
of items $S \subseteq \Omega$ that maximize the bidder's utility at prices~$p$. 
We call a set in the demand a \emph{demand representative}. Throughout the 
paper we omit subscripts if they are clear from the context.

\begin{definition}[Envy free]
An allocation $\langle \Omega_b\rangle$ of items $\Omega$ to bidders $\B$ is 
envy free \textup{(}on $\Omega$\textup{)} for some prices $p$ if for all bidders $b \in \B$, $\Omega_b\in D_b(p)$.
%and all sets of items $T \subseteq \Omega$ it holds that $u_{b,p}(\Omega_b) \ge u_{b,p}(T)$.
 We say that prices $p$ are envy free \textup{(}on $\Omega$\textup{)} if there 
exists an envy-free allocation \textup{(}on $\Omega$\textup{)} for these prices.
\end{definition}
There exist envy-free prices for any valuation functions of the 
bidders, e.g., set all prices to $\max_{b, S} v_b(S)$. For these prices
the allocation which does not allocate any item is envy free. 
Thus also minimal envy-free prices always exist, but are in general not unique.

\begin{definition}[Walrasian equilibrium \textup{(}\we{}\textup{)}]
A Walrasian equilibrium  \textup{(}on $\Omega$\textup{)} is an envy-free allocation $\langle 
\Omega_b\rangle$  \textup{(}on $\Omega$\textup{)} with prices $p$ such that all prices are non-negative 
and the price of unallocated items is zero. We call the allocation 
$\langle \Omega_b\rangle$ a Walrasian allocation  \textup{(}on $\Omega$\textup{)} 
and the prices $p$ Walrasian prices  \textup{(}on $\Omega$\textup{)}.
\end{definition}
We assume that the central seller has a value of zero for every subset of the 
items; thus (with quasi-linear utility functions) selling all items makes the 
seller envy free.
In this case a Walrasian equilibrium can be seen as an \emph{envy-free 
two-party equilibrium}, i.e., envy free for the buyers and the seller.
Note that for a Walrasian price vector there might exist multiple 
envy-free allocations.

\subsection{Valuation Classes}
\ignore{ %JOURNAL VERSION:
Monotone valuation functions are valuation functions for which for every 
$S \subseteq T \subseteq \Omega$ it holds that $v(S) \le v(T)$.
We explicitly do \emph{not} assume
monotonicity of the valuation functions in general, while many known results
rely on this property.}

A \emph{unit demand} valuation assigns a value to every item and defines the value
of a set as the \emph{maximum} value of an item in it.
An \emph{additive} valuation also assigns a value to every item but defines
the value of a set as the \emph{sum} of the values of the items in the set.
Non-negative unit demand and non-negative additive valuations both have the 
gross-substitute property (defined below) and are by definition monotone. 

\begin{definition}[Gross substitute \textup{(}\gs{}\textup{)}]
A valuation function is \emph{gross substitute} if for every two price vectors 
$p^{(2)} \ge p^{(1)} \ge 0$ and every set $D^{(1)} \in D(p^{(1)})$, there exists a set 
$D^{(2)} \in D(p^{(2)})$ with $j \in D^{(2)}$ for every $j \in D^{(1)}$ with $p^{(1)}_j = p^{(2)}_j$.
\end{definition}

For \emph{gross-substitute} valuations of the bidders a Walrasian equilibrium
is guaranteed to exist in a two-sided market~\cite{KelsoCr1982} and can be 
computed in polynomial time~\cite{NisanS2006,PaesLeme2014}. Further, gross 
substitute is the maximal valuation class containing the unit demand class for 
which the former holds~\cite{GS1999}. 
Several equivalent definitions are known for this 
class~\cite{GS1999,PaesLeme2014}. 
% The class is a subclass of \emph{sub-modular} valuations.
We will further use that for gross-substitute valuations the
Walrasian prices form a complete lattice~\cite{GS1999}.

We define next an $\OR$-valuation. 
Lehmann~et~al.~\cite{LehmannLeNi2006} showed that the $\OR$ of 
gross-substitute valuations is gross substitute.
\begin{definition}[$\OR$-player]\label{def:OR}
The $\OR$ of two valuations $v$ and $w$ is defined as 
$(v\;\OR\; w)(S) = \max_{R, T \subseteq S, R \cap T = \emptyset}(v(R) + w(T))$.
Given a set of valuations $\{v_b\}$ for bidders $ b \in \B$,
we say that the $\OR$-player is a player with valuation 
$v_{\OR}(S) = \max_{\langle S_b\rangle} 
\sum_{b \in \B} v_b(S_b)\,.$
\end{definition}
% When valuations are monotone, the $\OR$ can w.l.o.g.\ be defined using a partition instead of an allocation. 

\section{Model and Equilibrium} \label{sec:model}
There are $k$ items to be allocated to $m$ mediators. Each mediator $\M_i$
represents a set $\B_i$ of bidders, where $|\B_i|=n_i$. Each bidder
is connected to a unique mediator. Each bidder has a valuation function over 
all subsets of the items and a quasi-linear utility function.  
A \emph{central auction} is an 
auction run on all items with mediators as bidders. After an allocation $\langle \Omega_i\rangle$
and prices $r$ at 
the central auction are set, another $m$ \emph{local auctions} are conducted, one by 
each mediator. In the 
local auction for mediator $\M_i$ the items $\Omega_i$ that were allocated to 
her in the central auction are the sole
items and the bidders $\B_i$ are the sole bidders. A solution 
is an assignment of central-auction and local-auction prices to items and 
an allocation of items to bidders and hence,
by uniqueness, also to mediators. We define next a three-party equilibrium
based on envy-freeness.
\begin{definition}[Equilibrium]
A \emph{three-party competitive equilibrium} is an allocation of items to bidders
and a set of $m+1$ price vectors $r,p^1,p^2,\ldots,p^m$ such that the
following requirements hold.
For $1 \le i \le m$
\begin{enumerate}
	\item every mediator\footnote{Independent of how the demand of a mediator is 
	defined.} $\M_i$ is allocated a set $\Omega_i$ in her demand at price $r$,
	\item every item $j$ with non-zero price $r$ is allocated to a mediator,
	\item the price $p^i$ coincides with $r$ for all items not in $\Omega_i$\label{rq}, 
	\item and every bidder $b\in \B_i$ is allocated a subset of $\Omega_i$ that is in his demand at price $p^i$.
% 	\item and the price $p^i$ is a minimal envy-free price vector for the bidders in $\B_i$.
\end{enumerate}
\end{definition}
In other words, the allocation to the bidders in $\B_i$ with prices $p^i$ must be
envy-free for the bidders, the allocation to the mediators with
prices $r$ must be envy free for the mediators and for the central seller, i.e., 
must be a Walrasian equilibrium; and the prices $p^i$ must be equal 
to the prices $r$ for all items not assigned to mediator $\M_i$.

Note that the allocation of the items to the mediators and prices~$r$ are the 
outcome of a \emph{central} auction run by the central seller, while the allocation
to the bidders in $\B_i$ and prices~$p^i$ correspond to the outcome of a \emph{local}
auction run by mediator~$\M_i$. These auctions are connected by the demands
of the mediators and Requirement~\ref{rq}.

% \ignore{
We next present our mediator model.
The definition of an Envy-Free Mediator, or $\m$-mediator for short,
reflects the following idea:
To determine her revenue for a set of items $S$ at central auction prices $r$,
the mediator simulates the local auction she would run if she would obtain
the set $S$ at prices $r$. Given the outcome of this ``virtual auction'',
she can compute her potential revenue for $S$ and $r$ as the difference between the
virtual auction prices of the items sold in the virtual auction
and the central auction prices for the items in $S$.
However, as motivated in the introduction, the mediator is required to represent the preferences of her
bidders and therefore not every set $S$ is ``allowed'' for the mediator,
that is, for some sets the revenue of the mediator is set to $-1$.
The sets that maximize the revenue are then in the demand of the mediator at 
central auction prices $r$.
To make the revenue of a mediator well-defined and to follow our motivation
that a mediator should satisfy her bidders, the virtual auctions specifically
compute minimal envy-free price vectors.
% }

\begin{definition}[Envy-Free Mediator]\label{def:mediator}
An $\m$-mediator $\M_i$ determines her demand for a price vector $r \ge 0$ as follows. For 
each subset of items $S \subseteq \Omega$ she runs a virtual auction with items
$S$, her bidders $\B_i$,
and reserve prices $r$. We assume that the virtual auction computes minimal 
envy-free prices $p^S \ge r$ and a corresponding envy-free allocation $\langle S_b\rangle$.\inshort{\footnote{If there 
are multiple envy-free allocations on $S$ for the prices $p^S$, the mediator 
chooses one that maximizes $\sum_{b \in \B_i} p^S(S_b)$.}}
We extend the prices $p^S$ to all items in $\Omega$ by setting $p^S_j = r_j$
for $j \in \Omega \setminus S$,
and define the revenue $R_{i,r}(S)$ of the mediator for a set $S$ as follows.
If the allocation $\langle S_b\rangle$ is envy free for the bidders $\B_i$ and prices $p^S$ on 
$\Omega$, then $R_{i,r}(S) = \sum_{b \in \B_i} p^S(S_b) - r(S)$;
otherwise, we set $R_{i,r}(S) = -1$.\infull{\footnote{For the results of this paper this could
be any negative value including $-\infty$.}}
The demand $D_{i}(r)$ of $\M_i$ is the set
of all sets $S$ that maximize the revenue of the mediator for the reserve prices $r$.
The local auction of $\M_i$ for a set $\Omega_i$ allocated to her
in the central auction at prices $r$ is equal to her virtual auction for $\Omega_i$
and $r$.
\end{definition}

\infull{Note that for a set $S$ with $R_{i,r}(S) = \sum_{b \in \B_i} p^S(S_b) - r(S)$
the revenue of an $\m$-mediator $\M_i$ is maximal if the envy-free allocation 
on $S$ is such that $\sum_{b \in \B_i} p^S(S_b)$ is as high as possible. Thus if there 
are multiple envy-free allocations on $S$ for the prices $p^S$, the mediator 
chooses one that maximizes $\sum_{b \in \B_i} p^S(S_b)$.}

\ignore{\footnote{
We need this assumption in our proof that for \gs{} valuations
of the bidders a three-party competitive equilibrium exists. However,
in this case the assumption can easily be satisfied within a mechanism
that computes such an equilibrium by the way an $\m$-mediator can then determine
a set in her demand.}}
% we will only run the virtual auction on $\Omega$, take the set S of allocated items
% and then simply take this allocation as the allocation for S

Following the above definition, we say that a price vector
is \emph{locally envy free} if it is envy free for the bidders $\B_i$ on
the subset $\Omega_i \subseteq \Omega$ assigned to mediator $\M_i$ and 
\emph{globally envy free} if it is envy free for the bidders $\B_i$ on $\Omega$.
Note that if $p^S$ is envy free on $\Omega$, then it is minimal envy free 
$\ge r$ on $\Omega$ for the bidders $\B_i$.

An interesting property of $\m$-mediators is that every
Walrasian equilibrium in the central auction can be combined with the
outcome of the local auctions of $\m$-mediators to form a three-party
competitive equilibrium.

\begin{theorem}\label{th:mediator}
Assume all mediators are $\m$-mediators. Then a Walrasian equilibrium in the 
central auction with allocation $\langle \Omega_i \rangle$ together with 
the allocation and prices computed in the local auctions of the mediators $\M_i$
on their sets $\Omega_i$ \textup{(}not necessarily Walrasian\textup{)} form a three-party competitive equilibrium.
\end{theorem}

\begin{proof}
A Walrasian equilibrium in the central auction is a price vector $r \ge 0$ and
an allocation $\langle \Omega_i \rangle$ of items to mediators such that
every item with strictly positive price is allocated to a mediator and
every mediator is allocated a set in her demand $D_{i}(r)$.
By the definition of $D_{i}(r)$, the virtual auction for every
set $S \in D_{i}(r)$ computes an
allocation of the items in $S$ to her bidders and envy-free prices $p^i \ge r$
(on $\Omega$) such that every bidder in $\B_i$ is allocated a set in his demand at
prices $p^i$ and $p^i_j = r_j$ for all items $j \notin S$. Thus all requirements
of a three-party competitive equilibrium are satisfied.
\end{proof}

Further, with $\m$-mediators a three-party competitive equilibrium exists whenever a 
Walrasian equilibrium exists for the bidders and items without the mediators.
\begin{theorem}\label{th:exist}
Assume all mediators are $\m$-mediators and a Walrasian equilibrium exists for the 
set of bidders and items \textup{(}without mediators\textup{)}. Then there exists
a three-party competitive equilibrium.
\end{theorem}

\begin{proof}
A Walrasian equilibrium  is a price vector $r \ge 0$ and
an allocation $\langle \Omega_b \rangle$ of items to bidders such that
every bidder is envy-free and all items with non-zero price are allocated
to a bidder. This equilibrium induces a trivial three-party competitive 
equilibrium where all price vectors are identical to $r$ and the allocation to
mediators is uniquely determined by the allocation~$\langle \Omega_b \rangle$ 
to bidders. To see this note that the allocation $\langle \Omega_b 
\rangle$ with prices $r$ is globally envy-free for all bidders and thus for a 
mediator $\M_i$ the minimal locally envy-free prices $\ge r$ are equal to $r$ for 
the set of items allocated to $\B_i$. The revenue of all mediators under this 
equilibrium is zero and for each mediator the set allocated to her is in her 
demand.\footnote{The above proof also holds for 
any other mediator definition that prohibits mediators to gain other revenue  
than from the competition between her bidders in the local auction. This is 
because there is no competition in 
the local auction when the allocation and prices in the central auction are 
determined by a Walrasian equilibrium between bidders and items.}
\end{proof}

The proof of Theorem~\ref{th:exist} only shows the existence of trivial 
three-party equilibria
that basically ignores the presence of mediators. However, three-party equilibria
and $\m$-mediators allow for richer outcomes that permit the mediators to
gain revenue from the competition between their bidders while still representing
the preferences of their bidders towards the central seller. 
In the next section we show how to find such an equilibrium provided that the
valuations of all bidders are gross substitute. Recall that gross-substitute
valuations are the most general valuations that include unit demand valuations
for which a Walrasian equilibrium exists~\cite{GS1999}; and that efficient
algorithms for finding a Walrasian equilibrium are only known for this valuation class.

\section{An Efficient Algorithm for Gross-substitute Bidders}\label{sec:alg}

In this section we will show how to find, in polynomial time, a three-party 
competitive equilibrium if the valuations of all bidders are gross substitute. 
The prices the bidders have to pay at equilibrium, and thus the utilities 
they achieve, will be the same as in a Walrasian equilibrium (between bidders 
and items) with minimum prices 
(Section~\ref{sec:minWEprices}). The price
the bidders pay is split between the mediators and the exchange. We show
how to compute an equilibrium where this split is best for the mediators
and worst for the exchange. In turn the computational load can be split
between the mediators and the exchange as well.
The algorithm will be based on existing algorithms to compute Walrasian
equilibria for gross-substitute bidders.

The classical (two-party) \emph{allocation problem} is the following:
We are given $k$~items and $n$~valuation functions and we should find
an equilibrium allocation (with or without equilibrium prices) if one exists.
Recall that in general a valuation function has a description of size 
exponential in~$k$. Therefore, the input valuation functions can only be 
accessed via an \emph{oracle}, defined below. 
An \emph{efficient} algorithm runs in time polynomial in $n$ and $k$ (where 
the oracle access is assumed to take constant time).

Given an algorithm that computes a Walrasian allocation for gross-substitute 
bidders, by a result of Gul and Stacchetti~\cite{GS1999} minimum Walrasian 
prices can be computed by solving the allocation problem 
$k+1$~times. A Walrasian allocation can be combined with any 
Walrasian prices to form a Walrasian equilibrium~\cite{GS1999}.
Thus we can assume for gross-substitute valuations 
that a polynomial-time algorithm for the allocation problem also returns a
vector of minimum prices that support the allocation. 

Two main oracle definitions that were considered in the literature are the
\emph{valuation oracle}, where a query is a set of items~$S$ and the oracle
replies with the exact value of~$S$; and the \emph{demand oracle}, where a 
query is a price vector~$p$ and the oracle replies with a demand 
representative~$D$~\cite{AGT}.
Note that in the literature the answer of a demand oracle is sometimes defined 
to be all sets in the demand, however this cannot be assumed to be of
polynomial size even for gross-substitute valuations.

It is known that a demand oracle is strictly
stronger than a valuation oracle, i.e., a valuation query can be
simulated by a polynomial number of demand queries but not vice versa.
For gross-substitute valuations, however, these two query models are polynomial-time equivalent,
see Paes~Leme~\cite{PaesLeme2014}.
The two-party allocation problem is efficiently solvable for gross-substitute
valuations~\cite{NisanS2006,PaesLeme2014}.
For other valuations efficient algorithms are not known even in the demand
query model.

We define the \emph{three-party allocation problem} in the same manner.
We are given $k$ items, $n$ valuation functions over the subsets of items
and $m$ mediators,
each associated with a set of unique bidders. We are looking for
a three-party equilibrium allocation (and equilibrium prices) if one exists.
We will assume that the input valuations are given through a valuation oracle.
An efficient algorithm runs in time polynomial in $n$ and $k$ (hence also in $m \leq n$).

% We further assume all valuations are gross substitute.
The algorithm will be based on the following central result: For gross
substitute valuations of the bidders an $\m$-mediator
and an $\OR$-player over the valuations of the same bidders are 
equivalent with respect to their demand and their allocation of items to 
bidders. Thus in this case $\m$-mediators can be considered as if they have a 
gross-substitute valuation.
Note that for general valuations this equivalence does not hold.

\begin{theorem}\label{th:equal}
If the valuation functions of a set of bidders $\B_i$ are gross substitute,
then the demand of an $\m$-mediator for $\B_i$ is equal to the demand of an 
$\OR$-player for $\B_i$.
Moreover, the allocation in a virtual auction of the $\m$-mediator for reserve
prices $r$ and a set of items $S$ in the demand is an optimal allocation for the 
$\OR$-player for $S$ and $r$ and vice versa.
\end{theorem}

To this end, we will first show for the virtual (and local) auctions that a modified Walrasian 
equilibrium, the \newwe{}, exists for gross-substitute valuations with reserve prices.
For this we will use yet another reduction to a (standard) Walrasian equilibrium
without reserve prices but with an additional additive player\footnote{Such a player 
was introduced by Paes~Leme~\cite{PaesLeme2014} to find
the demand of an $\OR$-player (with a slightly different definition of $\OR$).}.

\begin{definition}[Walrasian equilibrium with reserve prices $r$
\textup{(}\newwe{}\textup{)}~\cite{GuruswamiHaKaKeKeMc2005}]
A Walrasian equilibrium with reserve prices $r \ge 0$ \textup{(}on $\Omega$\textup{)} is an 
envy-free allocation $\langle \Omega_b\rangle$ \textup{(}on $\Omega$\textup{)} with prices 
$p$ such that $p\geq r$, 
and the price of$\ $every unallocated item is equal to its reserve price,
i.e., $p_j = r_j$ for $j \not\in \cup_b \Omega_b$. We say that $\langle \Omega_b\rangle$
is a \newwe{} allocation \textup{(}on $\Omega$\textup{)} and $p$ are \newwe{} prices \textup{(}on $\Omega$\textup{)}.
\end{definition}

\subsection{Properties of Walrasian Equilibria with Reserve Prices}\label{sec:we}
In this section we generalize several results about Walrasian equilibria
to Walrasian equilibria with reserve prices. Similar extensions were shown
for unit demand valuations in~\cite{GuruswamiHaKaKeKeMc2005}.

We first define a suitable linear program. The \newlp{}, shown below, is a
linear program obtained from a reformulation of the dual of the LP-relaxation 
of the welfare maximization integer program after adding reserve prices $r \ge 0$. 
\newlength{\suml}
\settowidth{\suml}{$\scriptstyle b \in \B,\, S \mid j \in S$}
\newlength{\maximize}
\settowidth{\maximize}{subject to}
% \begin{align}
% \makebox[\maximize]{\text{maximize}}\quad 
% &\sum_{\makebox[\suml]{$\scriptstyle b \in \B,\, S \subseteq \Omega$}} 
% x_{b,S} v_b(S) + \sum_{j \in \Omega} \left(1-\sum_{b \in \B,\, S \mid j \in S} x_{b, S}\right) r_j &\tag{\newlp}\\
% \text{subject to}\quad 
% &\sum_{b \in \B,\, S \mid j \in S} x_{b, S} \le 1 \ \ \forall j \in \Omega\\
% &\sum_{\makebox[\suml]{$\scriptstyle S \subseteq \Omega$}} x_{b, S} \le 1 
% \ \ \forall b \in \B\\
% &\makebox[\suml]{} x_{b,S} \ge 0 \ \ \ \forall b \in \B,\, S\subseteq \Omega
% \end{align}
% \begin{IEEEeqnarray}{s"cl}
% 	maximize & \sum_{b \in \B,\, S \subseteq \Omega}
% 	& x_{b,S}  v_b(S) + \sum_{j \in \Omega} 
% 	\left(1-\sum_{b \in \B,\, S \mid j \in S} x_{b, S}\right) r_j \\
% 	\\
% 	subject to & \sum_{b \in \B,\, S \mid j \in S} &
% 	\begin{IEEEeqnarraybox}{RCl"l}
% 		x_{b, S} & \le & 1 & \forall j \in \Omega
% 	\end{IEEEeqnarraybox}\\
% 	& \sum_{S \subseteq \Omega} &
% 	\begin{IEEEeqnarraybox}{RCl"l}
% 		x_{b, S} & \le & 1 & \forall b \in \B
% 	\end{IEEEeqnarraybox}\\
% 	& &
% 	\begin{IEEEeqnarraybox}{RCl"l}
% 		x_{b,S} & \ge & 0 & \forall b \in \B,\, S\subseteq \Omega
% 	\end{IEEEeqnarraybox}
% \end{IEEEeqnarray}
\begin{IEEEeqnarray*}{0s"l}
	maximize & \sum_{b \in \B,\, S \subseteq \Omega}
	 x_{b,S}  v_b(S) + \sum_{j \in \Omega} 
	\left(1-\sum_{b \in \B,\, S \mid j \in S} x_{b, S}\right) r_j 
% 	\IEEEyesnumber \ztag{\newlp} % \label{test}
	\\
	\\
	subject to & 
	\begin{IEEEeqnarraybox}[][t]{cRCl"lL}
		\sum_{b \in \B,\, S \mid j \in S} & x_{b, S} & \le & 1 & \forall j & \in \Omega\\
		\sum_{S \subseteq \Omega} & x_{b, S} & \le & 1 & \forall b & \in \B\\
		& x_{b,S} & \ge & 0 & \forall b &\in \B,\, S\subseteq \Omega
	\end{IEEEeqnarraybox}
\end{IEEEeqnarray*}

We now show how the \newlp{} is obtained.
It is well known that for any collection $\{v\}$ of valuations a Walrasian
equilibrium (\we{}) exists if and only if the linear programming relaxation 
of the welfare maximization problem (\lp{}), given below, has an integral
solution. The integral solution combined with optimal dual prices yields a Walrasian 
equilibrium and vice versa (see e.g.~\cite{BikhchandaniMa1997} for monotone
valuations and~\cite{MWG1995} for more general valuations).
% \newlength{\suml}
% \settowidth{\suml}{$\scriptstyle b \in \B,\, S \mid j \in S$}
% \begin{align}\label{LP}
% \makebox[\maximize]{\text{maximize}}\quad 
% &\sum_{\makebox[\suml]{$\scriptstyle b \in \B,\, S \subseteq \Omega$}} x_{b,S} v_b(S) & \tag{\lp}\\
% \text{subject to}\quad 
% &\sum_{b \in \B,\, S \mid j \in S} x_{b, S} \le 1 &\forall j &\in \Omega\\
% &\sum_{\makebox[\suml]{$\scriptstyle S \subseteq \Omega$}} x_{b, S} \le 1 
% &\forall b &\in \B\\
% &\makebox[\suml]{} x_{b,S} \ge 0 &\forall b &\in \B,\, S\subseteq \Omega
% \end{align}

\begin{IEEEeqnarray*}{0s"l}
	maximize & \sum_{b \in \B,\, S \subseteq \Omega} x_{b,S} v_b(S) 
	\IEEEyesnumber \ztag{\lp} \label{LP}
	\\
	subject to &
	\begin{IEEEeqnarraybox}[][t]{cRCl"lL}
		\sum_{b \in \B,\, S \mid j \in S} & x_{b, S} & \le & 1 & \forall j & \in \Omega\\
		\sum_{S \subseteq \Omega} & x_{b, S} & \le & 1 & \forall b & \in \B\\
		& x_{b,S} & \ge & 0 & \forall b & \in \B,\, S\subseteq \Omega
	\end{IEEEeqnarraybox}
\end{IEEEeqnarray*}
The dual is as follows.
% \begin{align}
% \makebox[\maximize]{\text{minimize}}\quad 
% &\sum_{b \in \B} u_b + \sum_{j \in \Omega} p_j &\\
% \text{subject to}\quad 
% &u_b + \sum_{j \in S} p_j \ge v_b(S) &\forall b &\in \B,\, S \subseteq \Omega\\
% &u_b \ge 0,\, p_j \ge 0 &\forall b &\in \B,\, j \in \Omega
% \end{align}
\begin{IEEEeqnarray*}{0s"l}
	minimize & \sum_{b \in \B} u_b + \sum_{j \in \Omega} p_j \\
	subject to &
	\begin{IEEEeqnarraybox}[][t]{RCl"lL}
		u_b + \sum_{j \in S} p_j & \ge & v_b(S) & 
		\forall b & \in \B,\, S \subseteq \Omega\\
		u_b & \ge & 0 & \forall b & \in \B\\ 
		p_j & \ge & 0 & \forall j & \in \Omega
	\end{IEEEeqnarraybox}
\end{IEEEeqnarray*}
We will think of the dual variables $p_j$s as prices of items and of
$u_b$s as maximum utilities for the bidders. Note that the dual objective is
a function of the $p$s as the $u$s are determined by them.
Now consider the effect of reserve prices, i.e., for all $j \in \Omega$ a 
lower bound $r_j \ge 0$ for the dual variables $p_j$.
% \begin{align}\label{rDP}
% \makebox[\maximize]{\text{minimize}}\quad 
% &\sum_{b \in \B} u_b + \sum_{j \in \Omega} p_j &\\
% \text{subject to}\quad 
% &u_b + \sum_{j \in S} p_j \ge v_b(S) &\forall b &\in \B,\, S \subseteq \Omega\\
% &u_b \ge 0,\, p_j \ge r_j &\forall b &\in \B,\, j \in \Omega
% \end{align}
\begin{IEEEeqnarray*}{0s"l}
	minimize & \sum_{b \in \B} u_b + \sum_{j \in \Omega} p_j \\
	subject to &
	\begin{IEEEeqnarraybox}[][t]{RCl"lL}
		u_b + \sum_{j \in S} p_j & \ge & v_b(S) & 
		\forall b & \in \B,\, S \subseteq \Omega\\
		u_b & \ge & 0 & \forall b & \in \B\\ 
		p_j & \ge & r_j & \forall j & \in \Omega
	\end{IEEEeqnarraybox}
\end{IEEEeqnarray*}
We can reformulate this linear program by a variable transformation with $q_j = p_j - r_j$
for all $j \in \Omega$. The term $\sum_{j \in \Omega} r_j$ is part of the input and thus
can be omitted from the objective value.
% \begin{align}\label{qDP}
% \makebox[\maximize]{\text{minimize}}\quad 
% &\sum_{b \in \B} u_b + \sum_{j \in \Omega} q_j + \sum_{j \in \Omega} r_j&\\
% \text{subject to}\quad 
% &u_b + \sum_{j \in S} q_j \ge v_b(S) - \sum_{j\in S} r_j &
% \forall b &\in \B,\, S \subseteq \Omega\label{eq:dualEF}\\
% &u_b \ge 0,\, q_j \ge 0 &\forall b &\in \B,\, j \in \Omega
% \end{align}
\begin{IEEEeqnarray*}{0s"l}
	minimize & \sum_{b \in \B} u_b + \sum_{j \in \Omega} q_j 
	+ \sum_{j \in \Omega} r_j \\
	subject to &
	\begin{IEEEeqnarraybox}[][t]{RCl"lL}
		u_b + \sum_{j \in S} q_j & \ge & v_b(S) - \sum_{j\in S} r_j & 
		\forall b & \in \B,\, S \subseteq \Omega\\
		u_b & \ge & 0 & \forall b & \in \B\\ 
		q_j & \ge & 0 & \forall j & \in \Omega
	\end{IEEEeqnarraybox}
\end{IEEEeqnarray*}
With this reformulation we obtain the following primal, which we call
\newlp{}. Again $\sum_{j \in \Omega} r_j$ can be omitted from the objective value
without changing the set of solutions.
% \begin{align}\label{qLP}
% \makebox[\maximize]{\text{maximize}}\quad 
% &\sum_{\makebox[\suml]{$\scriptstyle b \in \B,\, S \subseteq \Omega$}} 
% x_{b,S} \left(v_b(S) - \sum_{j\in S} r_j\right) + \sum_{j \in \Omega} r_j &\tag{\newlp}\\
% \text{subject to}\quad 
% &\sum_{b \in \B,\, S \mid j \in S} x_{b, S} \le 1 &\forall j &\in \Omega\label{eq:primalitems}\\
% &\sum_{\makebox[\suml]{$\scriptstyle S \subseteq \Omega$}} x_{b, S} \le 1 
% &\forall b &\in \B\label{eq:primalbidders}\\
% &\makebox[\suml]{} x_{b,S} \ge 0 &\forall b &\in \B,\, S\subseteq \Omega
% \end{align}
\begin{IEEEeqnarray*}{0s"l}
	maximize & \sum_{b \in \B,\, S \subseteq \Omega}
	x_{b,S} \left(v_b(S) - \sum_{j\in S} r_j\right) + \sum_{j \in \Omega} r_j 
	\IEEEyesnumber \ztag{\newlp} \label{qLP}
	\\
	\\
	subject to & 
	\begin{IEEEeqnarraybox}[][t]{cRCl"lL}
		\sum_{b \in \B,\, S \mid j \in S} & x_{b, S} & \le & 1 & \forall j & \in \Omega\\
		\sum_{S \subseteq \Omega} & x_{b, S} & \le & 1 & \forall b & \in \B\\
		& x_{b,S} & \ge & 0 & \forall b &\in \B,\, S\subseteq \Omega
	\end{IEEEeqnarraybox}
\end{IEEEeqnarray*}
The objective value of the \newlp{} can be rewritten as
\begin{equation}\label{eq:add}
\sum_{b \in \B,\, S \subseteq \Omega}
x_{b,S} v_b(S) + \sum_{j \in \Omega} \left(1-\sum_{b \in \B,\, S \mid j \in S} x_{b, S}\right) r_j \,.
\end{equation}

For an integral solution to the \newlp{} we can interpret this reformulation
as a solution to a \lp{} with an additional additive player whose value for 
an item  is equal to that item's reserve price. We will use this interpretation 
to extend known results for Walrasian equilibria to Walrasian equilibria with
reserve prices. The results are summarized in Theorem~\ref{th:reserve} below.
We use the following definition.
\begin{definition}[additional additive player]\label{def:additive}
Let $\{v_b\}$ be a set of valuation functions over $\Omega$ for 
bidders $b \in \B$, and let $r \ge 0$ be reserve prices for the items 
in $\Omega$. Let $\{v'_{b'}\}$ be the set of valuation functions when 
an additive bidder~$a$ is added, i.e., for the bidders 
$b' \in \B' = \B \cup \set{a}$ with $v'_{b'}(S) = v_{b'}(S)$ for ${b'} \ne a$ 
and $v'_a(S) = \sum_{j \in S} r_j$ for all sets $S \subseteq \Omega$.
For an allocation $\langle \Omega_b \rangle_{b \in \B}$ we define $\langle 
\Omega'_{b'} \rangle_{b' \in \B'}$ with $\Omega'_{b'} = \Omega_{b'}$ for 
${b'} \ne a$ and $\Omega'_a = \Omega \setminus \cup_b \Omega_b$.
\end{definition}

Theorem~\ref{th:reserve} will be used in the next section to characterize the 
outcome of the virtual auctions of an $\m$-mediator. It also provides a 
polynomial-time algorithm to compute a \newwe{} when the bidders in~$\B$ have 
gross-substitute valuations, given a polynomial-time algorithm for a \we{} for 
gross-substitute bidders.
\begin{theorem}\label{th:reserve}
\textup{(}a\textup{)} 
The allocation $\langle \Omega_b \rangle$ and the prices $p$ are a \newwe{} 
for $r \ge 0$ and bidders~$\B$ if and only if the allocation 
$\langle \Omega'_{b'} \rangle$ and prices $p'$ are a \we{} for the bidders~$\B'$, where we have $p_j = p'_j$ for $j \in \cup_{b \in \B} \Omega_b$ and 
$p_{j'} = r_{j'}$ for ${j'} \in \Omega \setminus \cup_{b \in \B} \Omega_b$
\textup{(}a1\textup{)}.
The allocation $\langle \Omega_b \rangle$ is a \newwe{} allocation if and only
if $\langle \Omega_b \rangle$ is an integral solution to the \newlp{} \textup{(}a2\textup{)}.

\textup{(}b\textup{)} If the valuations~$\{v\}$ are gross 
substitute, then \textup{(}b1\textup{)} there exists a \newwe{} for~$\{v\}$
and \textup{(}b2\textup{)} the \newwe{} price vectors form a complete lattice.
\end{theorem}

\begin{proof}[Proof of Theorem~\ref{th:reserve}]
(a1) $\Rightarrow$: Let $\langle \Omega_b \rangle$ and prices $p$ be a 
\newwe{} for bidders~$\B$. Then $\langle \Omega_b \rangle$ is an envy-free
allocation at prices $p \ge r$, and all unallocated items $j$ have price 
$p_j = r_j$. Let $\Omega_0$ denote the set of unallocated items. 
A \we{} for the bidders $\B'$ is given by prices $p$ and allocation 
$\langle \Omega'_{b'} \rangle$ with $\Omega'_{b'} = \Omega_{b'}$ for $b' \ne a$
and $\Omega'_a = \Omega_0$. All items are allocated in $\langle \Omega'_{b'} 
\rangle$. The allocation for the bidders $b' \ne a$ clearly is envy-free as
neither allocation nor prices were changed. Bidder $a$ is envy-free because 
$p \ge r$ and $p_j = r_j$ for $j \in \Omega'_a$.

$\Leftarrow$: Let $\langle \Omega'_{b'} \rangle$ and prices $p'$ be a \we{}
for the bidders~$\B'$. Then $\langle \Omega'_{b'} \rangle$ is an envy-free
allocation and all unallocated items have a price of zero. For bidder $a$ to
be envy-free it must hold that all items~$j$ not allocated to~$a$ have a 
price $p_j \ge r_j$ and all items~$j'$ allocated to~$a$ have $p_{j'} \le r_{j'}$.
We construct a \newwe{} for the bidders $\B$ as follows. For all items allocated
to bidders in $\B$ in $\langle \Omega'_{b'} \rangle$ allocation and prices remain
the same. For all other items $j$ their price is set to $r_j$ and they are left
unallocated. The allocation for the bidders $\B$ remains envy-free because
the prices of the now unallocated items were only increased.

(a2): First note that for bidders $\B'$ in a \we{}, and therefore in an 
integral solution to the \lp{}, we can assume w.l.o.g.\ that all items are
allocated because we have $r \ge 0$ and therefore all otherwise unallocated
items can be allocated to the additive player~$a$.
The objective value of the \lp{} for an integral solution 
$\langle \Omega'_{b'} \rangle$ for bidders $\B'$
can be written as $\sum_{b \in \B} v_b(\Omega_b) + r(\Omega'_a)$, which
is, w.l.o.g., equal to $\sum_{b \in \B} v_b(\Omega_b) + r(\Omega \setminus \cup_{b 
\in \B} \Omega_b)$. The latter is equivalent to Equation~\eqref{eq:add} for 
the allocation $\langle \Omega_{b} \rangle$. Thus there is (w.l.o.g.) 
a one-to-one correspondence between integral solutions to the \lp{} for 
bidders~$\B'$ and integral solutions to the \newlp{} for bidders~$\B$. Hence,
$\langle \Omega'_{b'} \rangle$ is an optimal solution to the \lp{} for $\B'$
if and only if $\langle \Omega_{b} \rangle$ is an optimal solution to the
\newlp{} for $\B$. Note that the corresponding constraints are satisfied as both
$\langle \Omega'_{b'} \rangle$ and $\langle \Omega_{b} \rangle$ are allocations, 
respectively. To complete the proof, consider the following chain of ``iff'' statements.
{
% \small
\begin{align*}
(\langle \Omega_b \rangle, p) \text{ is a \newwe{} for } \B &\Longleftrightarrow 
(\langle \Omega'_b \rangle, p') \text{ is a \we{} for } \B' \,,\\
(\langle \Omega'_b \rangle, p') \text{ is a \we{} for } \B' 
&\Longleftrightarrow 
(\langle \Omega'_b \rangle) \text{ solves \lp{} for } \B'\,,\\
(\langle \Omega'_b \rangle) \text{ solves \lp{} for } \B'
&\Longleftrightarrow 
(\langle \Omega_b \rangle) \text{ solves \newlp{} for } \B\,,\\
\text{and thus} &\\
(\langle \Omega_b \rangle, p) \text{ is a \newwe{} for } \B &\Longleftrightarrow 
(\langle \Omega_b \rangle) \text{ solves \newlp{} for } \B\,.
\end{align*}}

(b1): The valuations $\set{v}$ of the bidders in $\B$ are gross substitute
if and only if the valuations $\set{v'}$ of the bidders in $\B'$ are gross 
substitute, as the only difference between $\B$ and $\B'$ is the additive
bidder $a$ whose value for an item $j$ is equal to its reserve price $r_j \ge 0$.
Recall that every (non-negative) additive valuation is gross substitute.
The claim then directly follows from (a1) and the existence of a \we{} 
for $\set{v'}$.

(b2): To show that the \newwe{} price vectors form a complete lattice, we
have to show that for any two \newwe{} price vectors $p_1$ and $p_2$ the
price vectors $\min(p_1, p_2)$ and $\max(p_1, p_2)$, where the $\min$ and the 
$\max$ is meant element-wise, are \newwe{} price vectors as well. We will
use (a1) and that for gross-substitute valuations \we{} price vectors form
a complete lattice. The latter implies that for two \we{} price vectors $q'_1$ 
and $q'_2$, we have that $q'_\text{min} = \min(q'_1, q'_2)$ and $q'_\text{max} =
\max(q'_1, q'_2)$ are \we{} price vectors as well. Recall the relation of 
$p'$ and $p$ in (a1), i.e., $p = \max(p', r)$. By (a1) we 
have that (i) $p'_1$ and $p'_2$ are \we{} price vectors for $\B'$
and (ii) $q_\text{min}$ and $q_\text{max}$ are \newwe{} price vectors for $\B$.
Let $q'_1 = p'_1$ and let $q'_2 = p'_2$. The claim follows from 
$\min(p_1, p_2) = q_\text{min}$ and $\max(p_1, p_2) = q_\text{max}$.
\end{proof}

\subsection{The Equivalence of the EF-mediator and the OR-player for 
Gross-substitute Valuations}\label{sec:equal}

In this section we prove Theorem~\ref{th:equal}, that is,
the equivalence, for gross-substitute bidders, between the demand of 
an $\m$-mediator $\M_i$ and the demand of an $\OR$-player and, for the sets in 
the demand, the equivalence of the \emph{allocations}
of items to bidders of an $\OR$-player and an $\m$-mediator in the sense that
the allocation implied by the $\OR$-player could be used
by the $\m$-mediator and vice versa. 

If the valuations of the bidders in $\B_i$ are
all  gross substitute, by Theorem~\ref{th:reserve}~(b) a \newwe{} 
with minimum prices exists for all reserve prices $r\ge 0$. 
We will use this several times in this section.
We phrase all the statements in this section for  
gross-substitute valuations of the bidders,
although they all hold as long as all minimal envy-free prices that respect 
the reserve prices are equal to the minimum \newwe{} prices.

The proof proceeds as follows. We first characterize the demand of an 
$\m$-mediator for bidders with  gross-substitute valuations. As a first
step we show that for such bidders an $\m$-mediator actually computes a
\newwe{} with minimum prices in each of her virtual auctions. The minimality
of the prices implies that whenever
the virtual auction prices for an item set $S$ are globally envy-free, they are 
also minimum \newwe{} prices for the set of all items $\Omega$ and the bidders in $\B_i$. Thus, given reserve 
prices~$r$, all virtual auctions of an $\m$-mediator result in the same 
price vector $p$ as long as they are run on a set $S$ with non-negative 
revenue. With the help of some technical lemmata, we then completely 
characterize the demand of an $\m$-mediator and show that the mediator
does not have to run \emph{multiple} virtual auctions to determine her demand;
it suffices to run \emph{one} virtual auction on $\Omega$ where the set
of allocated items is a set in the demand of the $\m$-mediator. 
Thus for gross-substitute bidders the mediator can efficiently 
answer demand queries and compute the outcome of her local auction.

Finally we compare the utility function of the $\OR$-player to the optimal
value of the \newlp{} to observe that they have to be equal (up to an 
additive constant) for item sets that are in the demand of the $\OR$-player. 
Combined with the above characterization 
of the demand of the mediator, we can then relate both demands at central 
auction prices $r$ to optimal solutions of the \newlp{} for $r$ and $\Omega$ 
and hence show the equality of the demands for these two 
mediator definitions for  gross-substitute valuations of the bidders.
Recall that an $\OR$-player over gross-substitute valuations has a 
gross-substitute valuation~\cite{LehmannLeNi2006}. Thus in this case we can 
regard the $\m$-mediator as having a gross-substitute 
valuation. This implies that a Walrasian equilibrium for the central 
auction exists and, with the efficient demand oracle defined above, can be 
computed efficiently when all bidders have gross-substitute valuations and all 
mediators are $\m$-mediators.

We start the proof of Theorem~\ref{th:equal} with showing that the mediator 
computes in every virtual auction a \newwe{} with minimum prices.
The proof of this lemma is given in the next subsection.
\begin{lemma}
\label{lem:efminef}
If the valuations of an $\m$-mediator's bidders are  gross
substitute, then the $\m$-mediator computes minimum \newwe{} prices 
in her virtual auctions, i.e., items not allocated in a virtual auction
have a price equal to their reserve price.
\end{lemma}

This lemma implies that whenever for a set of items $S$ a virtual auction computes 
globally envy-free prices~$p^S$, these prices have
to be equal to the minimum \newwe{} prices on $\Omega$.
\begin{corollary}\label{cor:priceeq}
If the valuation functions of all bidders $b \in \B_i$ are  gross 
substitute,
then for reserve prices $r\ge 0$ and all sets $S \subseteq \Omega$ such that 
$R_{i,r}(S) \ne -1$ the virtual auction prices $p^S$ are equal to 
$p^{\Omega}$ for all items in $\Omega$.
\end{corollary}

It follows that an item $j$ with $p^{\Omega}_j > r_j$ must be in all sets with
$R_{i,r}(S) \ne -1$ and thus in all demand representatives of the mediator. 
This implies that two sets $S$ and $S'$ with $R_{i,r}(S) \ne -1$
and $R_{i,r}(S') \ne -1$ can only differ in items $j$ with $p^{\Omega}_j = r_j$.
Thus if for both $S$ and $S'$ all items are allocated in the virtual auction,
then $R_{i,r}(S) = p^S(S) - r(S) = p^{S'}(S') - r(S') = R_{i,r}(S')$.
Furthermore if for a set $S''$ with $R_{i,r}(S'') \ne -1$ an item $j \in S''$
with $r_j > 0$ is not allocated in the virtual auction, then 
$R_{i,r}(S'') < R_{i,r}(S)$.
Hence, if for a set $S$ with $R_{i,r}(S) \ne -1$ all items
are allocated in the virtual auction, then $R_{i,r}(S) = \max_{S'} R_{i,r}(S')$ 
and thus $S$ is in the demand of the mediator. Note that by
Definition~\ref{def:mediator}, if there are multiple \newwe{} allocations 
on $S$ for the prices $p^S$, the mediator chooses the one that maximizes 
$\sum_{b \in \B_i} p^S(S_b)$, i.e., if the mediator can allocate all items in $S$,
she will.
\begin{corollary}\label{cor:SinD}
Assume that the valuation functions of all bidders $b \in \B_i$ are  gross 
substitute. Let $r \ge 0$ be some reserve prices.
If for some set $S$ with $R_{i,r}(S) \ne -1$ all items with strictly positive
reserve price can be allocated in the virtual auction of the mediator,
% i.e., the mediator allocates all these items to maximize $p^S(\cup S_b)$, 
then $S$ is in $D_i(r)$.
\end{corollary}

To completely characterize demand and allocation of the $\m$-mediator,
we first show a useful technical result.
We compare the minimum \newwe{} prices for a set 
$T \subseteq \Omega$ with the minimum \newwe{} prices for a subset 
$S \subseteq T$. For this we will use the following well-known 
result for Walrasian equilibria by Gul and Stacchetti~\cite{GS1999} that by 
Theorem~\ref{th:reserve}~(a1) also holds with reserve prices.
\begin{lemma}[\cite{GS1999}]\label{lem:combine}
Any Walrasian price vector combined with any Walrasian allocation yields a
Walrasian equilibrium.
\end{lemma}
\begin{corollary}[of Lemma~\ref{lem:combine} and Theorem~\ref{th:reserve}]
\label{cor:combine}
For $r \ge 0$ a \newwe{} price vector combined with any \newwe{} allocation 
yields a \newwe{}.
\end{corollary}

The following lemma shows that, for suitable sets $S$ and $T$ with $S \subseteq T$,
the minimum prices in a \newwe{} on $S$ are equal for items in $S$ to the
corresponding prices in $T$. 
Part~(a) of the lemma was shown for monotone gross-substitute 
valuations without reserve prices in~\cite{GS1999}.
\begin{lemma}\label{lem:lessitems}
Assume the valuation functions of all bidders $b \in \B_i$ are  gross 
substitute.
Let $T \subseteq \Omega$ be a set of items and let $S$ be a subset of $T$.
For fixed reserve prices $r \ge 0$, let 
$(\langle T_b \rangle, p^T)$ be a \newwe{} with minimum prices on $T$ and let $(\langle S_b \rangle, p^S)$ be a \newwe{} with minimum prices on $S$.
Then \textup{(}a\textup{)} $p^T_j \le p^S_j$ for all $j \in S$ and
\textup{(}b\textup{)} if $\cup_b T_b \subseteq S$, then $p^T_j = p^S_j$ for all 
$j \in S$ and $(\langle T_b \rangle, p^S)$ is a \newwe{} with minimum prices 
on $S$.
\end{lemma}
\begin{proof}
(a) Let $V$ be the maximal valuation of any bidder, i.e., 
$\max_{b, T' \subseteq \Omega} v_b(T')$. Let $p'_j = p^S_j$ for $j \in S$ 
and let $p'_j = \max(V, r_j)$ for $j \in T \setminus S$. 
Then $(\langle S_b \rangle, p')$ 
is envy-free for all bidders on $T$ and $p' \ge r$. 
By Lemma~\ref{lem:minef} the prices $p^T$ are the minimum envy-free 
prices $ \ge r$ on $T$. Thus $p^T_j \le p'_j$ 
for all $j \in T$ and hence $p^T_j \le p^S_j$ for all $j \in S$.

(b) If the set $S$ contains all items in $\cup_b T_b$, then the prices $p^T$
restricted to the set $S$ with the allocation $\langle T_b \rangle$ are
a \newwe{} on $S$. Thus by the minimality of the prices $p^S$, we have
$p^T_j \ge p^S_j$ for all $j \in S$. Combined with (a) this shows $p^T_j = p^S_j$ 
for all $j \in S$.

The allocation $\langle T_b \rangle$ with prices $p^T$ restricted to $S$
are a \newwe{} on $S$ and the prices $p^T$ restricted to the set $S$ are 
equal to the minimum \newwe{} prices $p^S$ on $S$. 
Hence by Corollary~\ref{cor:combine} $(\langle T_b \rangle, p^S)$ 
is a \newwe{} with minimum prices on $S$.
\end{proof}

To characterize the demand of the mediator, we further need that the maximum
revenue the mediator can obtain is non-negative for all reserve prices $r \ge 0$.
To compare the demand $D_i(r)$ of an $\m$-mediator to the demand of an 
$\OR$-mediator, we further use that for \emph{every} set in $D_i(r)$ all items with 
positive reserve price are allocated in the local auction.
\begin{lemma}\label{lem:revenue}\label{lem:partition}
Assume the valuation functions of all bidders $b \in \B_i$ are gross 
substitute and let $r \ge 0$ be any reserve price vector. 
\upbr{a} There exists a \upbr{potentially empty} set $S\subseteq \Omega$ such that 
the revenue $R_{i,r}(S)$ of an $\m$-mediator $\M_i$ is non-negative.
\upbr{b} For a set $T \in D_i(r)$ with virtual auction allocation
$\langle T_b \rangle$ all items $j \in T$ with $r_j > 0$ are allocated.
\end{lemma}
\begin{proof}
(a) Let $(\langle \Omega_b \rangle, p)$ be the outcome of the virtual auction of an 
$\m$-mediator for $\Omega$. Take $S = \cup_b \Omega_b$. By Lemma~\ref{lem:efminef} and 
Lemma~\ref{lem:lessitems}~(b) $(\langle \Omega_b \rangle, p)$ is not only envy-free 
on $\Omega$ but further is a \newwe{} with minimum prices for 
the virtual auction of an $\m$-mediator for the set $S$.
Thus the mediator can allocate all items in $S$ in her virtual auction for $S$. 
Thus by $p \ge r$ the revenue $R_{i,r}(S) = \sum_b p(\Omega_b) - r(S) = p(S) - r(S)$ 
of the mediator for the set $S$ is non-negative.

(b) By (a) we have $R_{i,r}(T) \ge 0$ and thus $R_{i,r}(T) = \sum_b p^T(T_b) - r(T)$. Consider the set $T' = \cup_b T_b$. Assume by contradiction some items with 
$r_j > 0$ are not allocated in $\langle T_b \rangle$.
By Lemma~\ref{lem:efminef} $p^T_j = r_j$ for all items $j \in \Omega \setminus T'$. For the virtual auction prices $p^{T'}$ for $T'$ we have by definition
$p^{T'}_j = r_j$ for $j \in \Omega \setminus T'$.
By Lemma~\ref{lem:lessitems}~(b) $p^{T'}_j = p^T_j$ for all $j \in T'$ and thus
$p^{T'}_j = p^T_j$ for all $j \in \Omega$. Thus $(\langle T_b \rangle, p^{T'})$ 
is envy-free on the whole set of items $\Omega$, i.e., $R_{i,r}(T') \ne -1$. 
The mediator can allocate all items in $T'$; hence, $R_{i,r}(T') = p(T') - r(T') 
> R_{i,r}(T) = p(T') - r(T)$, a contradiction to $T \in D_i(r)$.
\end{proof}
This proof gives us immediately an efficient way to determine a set $S$
with $R_{i,r}(S) \ge 0$: Run the virtual auction on $\Omega$ with reserve prices $r$
and return the set $S$ of allocated items. Combined with 
Corollary~\ref{cor:priceeq}, this procedure
actually yields not only a set with non-negative revenue but even a set in the 
demand of the mediator.

Before we continue, we observe a relation
between the utility of the $\OR$-player for reserve prices $r$ and the $\OR$
over modified valuation functions\footnote{The valuations $\set{\widetilde{v}}$
might be non-monotone even if the valuations $\set{{v}}$ are monotone. This is not relevant here.} $\set{\widetilde{v}}$ with $\widetilde{v}_b(S) = 
v_b(S) - r(S)$ for all $S \subseteq \Omega$. Note that $\widetilde{v}_{\OR}(S) + 
r(S)$ equals the optimal value of the \newlp{} on $S$
as long as an optimal integral solution exists.
% By Lemma~\ref{lem:efminef}
% and Theorem~\ref{th:reserve}~(a2) an $\m$-mediator computes an optimal integral
% solution to the \newlp{} for \gs{} valuations of the bidders.
This relation gives a characterization of the demand of the $\OR$-player 
with reserve prices.
\begin{observation}\label{obs:or}
The utility of the $\OR$-player at prices $r$ is given by
$$u_{\OR,r}(S) = \max_{\langle S_b\rangle} \left( 
\sum_{b \in \B_i} v_b(S_b) \right) - r(S) = \max_{\langle S_b\rangle} \left( 
\sum_{b \in \B_i} v_b(S_b) - r(S_b) \right) - r(S \setminus \cup_b S_b)\,$$
The $\OR$ of the valuation functions $\widetilde{v}_b(S) = v_b(S) - r(S)$ 
is given by
$$\widetilde{v}_{\OR}(S) = \max_{\langle S_b \rangle} \sum_{b \in \B_i}
\widetilde{v}_b(S_b) =  \max_{\langle S_b\rangle} \left( \sum_{b \in \B_i} 
v_b(S_b) - r(S_b)\right)\,$$
By definition we have $\widetilde{v}_{\OR}(S) \ge u_{\OR,r}(S)$ 
\textup{(}1\textup{)}.\\
Let the allocation $\langle S^*_b\rangle$ be $\argmax_{\langle S_b \rangle} 
\sum_{b \in \B_i} \widetilde{v}_b(S_b)$ for the set $S$ and let 
$S^* = \cup_b S^*_b \subseteq S$. Then $\widetilde{v}_{\OR}(S) = \widetilde{v}_{\OR}
(S^*) = u_{\OR,r}(S^*)$ \textup{(}2a\textup{)}. Thus $\widetilde{v}_{\OR}(S) > 
u_{\OR,r}(S)$ iff $u_{\OR,r} (S^*) > u_{\OR,r}(S)$ iff $S \notin D_{\OR}(r)$ \textup{(}2b\textup{)}.
\end{observation}

The following two lemmata finally show that the demand of an $\m$-mediator
is equal to the demand of an $\OR$-player for any central auctions prices $r \ge 0$
and  gross-substitute valuations of the bidders. The proofs combine the
results obtained so far to relate both demands to an optimal solution of
the \newlp{} for reserve prices $r$ and the items in $\Omega$.
\begin{lemma}\label{lem:MtoOR}
If the valuation functions of all bidders $b \in \B_i$ are  gross 
substitute, then for any reserve prices 
$r \ge 0$ every set~$S$ in the demand of an $\m$-mediator is in the demand of the $\OR$-player.
Additionally, the $\OR$-player could use the $\m$-mediator's
allocation of the items in $S$ to the bidders in $\B_i$ to maximize her utility.
\end{lemma}
\begin{proof}
Let $S$ be a set in the demand $D_i(r)$ of an $\m$-mediator $\M_i$ for some reserve 
prices $r \ge 0$. By Lemma~\ref{lem:revenue}~(a) there exists a set $S'$ with
$R_{i,r}(S') \ge 0$, thus for $S$ in the demand we have $R_{i,r}(S) \ge 
R_{i,r}(S') \ge 0$. Let $(\langle S_b \rangle, p)$ be the outcome of the virtual 
auction of $\M_i$ for the set $S$. By Lemma~\ref{lem:efminef} 
$(\langle S_b \rangle, p)$ is a \newwe{} with minimum prices on 
the item set $S$. Since $R_{i,r}(S) \ge 0$, the
allocation $\langle S_b \rangle$ is envy-free on $\Omega$. As
the prices $p$ are \newwe{} prices for $S$ and are extended with $p_j = r_j$ for 
$j \in \Omega \setminus S$, the allocation $\langle S_b \rangle$ and prices $p$
are also a \newwe{} for the item set $\Omega$. Hence by Theorem~\ref{th:reserve}~(a2) 
the allocation $\langle S_b \rangle$ is an integral solution to the \newlp{} and 
thus maximizes the objective value of the \newlp{} for both the item sets $S$ 
and $\Omega$. Since the value of the \newlp{} only depends on the allocated sets,
the two objective values are the same. Note that by the definition of 
$\widetilde{v}_{\OR}$ in Observation~\ref{obs:or} the objective value of the 
\newlp{} is given by $\widetilde{v}_{\OR}(\Omega) + r(\Omega)$ with 
$\widetilde{v}_{\OR}(\Omega) = \max_{S' \in \Omega}\widetilde{v}_{\OR}(S')$. 
Thus we have $\widetilde{v}_{\OR}(S) = \widetilde{v}_{\OR}(\Omega) = \max_{S' \in \Omega}\widetilde{v}_{\OR}(S')$.
By Lemma~\ref{lem:partition}~(b) we can assume that all items $j \in S$ with 
$r_j > 0$ are allocated in $\langle S_b \rangle$, which implies that
$\sum_{b \in \B_i} r(S_b) = r(S)$ and thus $\widetilde{v}_{\OR}(S) = u_{\OR,r}(S)$. 
Since we have $\widetilde{v}_{\OR}(S') \ge u_{\OR,r}(S')$ 
for all $S' \subseteq \Omega$ by Observation~\ref{obs:or}~(1), this implies 
$u_{\OR,r}(S) \ge \max_{S' \in \Omega} u_{\OR,r}(S')$. Since $S \subseteq \Omega$,
it also holds that $u_{\OR,r}(S) \le \max_{S' \in \Omega} u_{\OR,r}(S')$, implying
that $u_{\OR,r}(S) = \max_{S' \in \Omega} u_{\OR,r}(S')$. Thus, $S$ is in the 
demand $D_{\OR}(r)$ of the $\OR$-player for reserve prices $r$ and the 
$\OR$-player could use the allocation $\langle S_b \rangle$ to maximize her utility.
\end{proof}

\begin{lemma}\label{lem:ORtoM}
If the valuation functions of all bidders $b \in \B_i$ are gross 
substitute, then for any reserve prices $r \ge 0$
every set~$S$ in the demand of the $\OR$-player is in the demand of an $\m$-mediator.
Additionally, the $\m$-mediator could use the
$\OR$-player's allocation of the items in $S$ to the bidders in $\B_i$ to maximize 
his revenue.
\end{lemma}
\begin{proof}
Let $S$ be a set in the demand $D_{\OR}(r)$ of the $\OR$-player for some reserve 
prices $r \ge 0$. Let $\langle S_b \rangle$ be the allocation of the $\OR$-player for the set $S$. By Observation~\ref{obs:or}~(2b) we have $u_{\OR,r}(S)
= \widetilde{v}_{\OR}(S)$. Recall that $\widetilde{v}_{\OR}(S) + r(S)$ is
equal to the objective value of the \newlp{} for the set $S$. 
Furthermore $\widetilde{v}_{\OR}(S) = \widetilde{v}_{\OR}(\Omega)$ because otherwise by 
Observation~\ref{obs:or}~(2a) there would be some allocation 
$\langle \Omega_b \rangle$ with $S' = \cup_b \Omega_b$ s.t.\ $\widetilde{v}_{\OR}(\Omega)
= \widetilde{v}_{\OR}(S') = u_{\OR,r}(S')$ and thus the utility of the
$\OR$-player for the set $S'$ would be higher than for the set $S$, contradicting
$S \in D_{\OR}(r)$. Hence the allocation $\langle S_b \rangle$ of the 
$\OR$-player is an integral solution to the \newlp{} 
on $S$ as well as on~$\Omega$. 
Let $p$ be the minimum \newwe{} prices~$p$
such that $(\langle S_b \rangle, p)$ is a \newwe{} on~$\Omega$. 
By Lemma~\ref{lem:lessitems}~(b) we know that $(\langle S_b \rangle, p)$, 
with the prices~$p$ restricted to~$S$, is also a \newwe{} with
minimum prices on~$S$. 
By Lemma~\ref{lem:efminef} the 
virtual auction of an $\m$-mediator $\M_i$ for the set~$S$ computes the same
unique minimum prices~$p$. 
Further $\widetilde{v}_{\OR}(S) = u_{\OR,r}(S)$ implies that all items with
strictly positive reserve price are allocated in $\langle S_b \rangle$.
Thus $\M_i$ could allocate all items in~$S$ with strictly positive reserve price
by using the allocation $\langle S_b \rangle$. The allocation $\langle S_b \rangle$
is envy-free at prices $p$ on $S$. Thus the revenue of the 
mediator for the set~$S$ is not set to~$-1$. Hence by Corollary~\ref{cor:SinD} 
the set~$S$ is in the demand of the $\m$-mediator.
\end{proof}

\subsubsection{Proof of Lemma~\ref{lem:efminef}} \label{sec:proofminef}
Lemma~\ref{lem:efminef} is a corollary to the following, more general, lemma.
\begin{lemma}\label{lem:minef}
Consider all envy-free outcomes with prices $p \ge r$ for a set of 
valuations $\{v\}$ and reserve prices $r \ge 0$.
If the valuations $\{v\}$ are gross substitute,
then the minimum \newwe{} prices are minimum envy-free prices $p$ with 
$p \ge r$ among all envy-free outcomes with $p \ge r$.
\end{lemma}
\begin{proof}
By Theorem~\ref{th:reserve}~(b2) all \newwe{} price vectors form a lattice. Let 
$p^*$ be the minimum price in this lattice. Recall that every \newwe{} price vector
is also an envy-free price vector. Assume by contradiction there exists 
an envy-free price $p$ such that $p^* \not\leq p$. Let $J = \{j \mid p_j < p^*_j\}$,
let $\delta = \min_{j\in J}{\{p^*_j - p_j\}}$ be the \emph{min-gap} and
$p^{*-\delta J} = p^* - \delta_J$ where $\delta_J$ is
the vector with value $\delta$ to each item in $J$ and $0$ otherwise. Note that by 
assumption $J\neq \emptyset,\ \delta>0$ and by minimality of $p^*$ no \newwe{} allocation 
exists for $p^{*-\delta J}$.

Let $w_b(q)$ denote the maximum utility of a bidder $b$ for a price
vector $q$, i.e., $w_b(q) = u_{b,q}(D)$ for some $D\in D_b(q)$.
Following Gul and Stacchetti~\cite{GulSt2000} and Ben-Zwi~et~al.~\cite{Ben-ZwiLaNe2013}, we define a requirement function 
and use Ben-Zwi~et~al.~\cite{Ben-ZwiLaNe2013}'s extension of (one direction of) Hall's Theorem.

\begin{definition}[requirement function]
Define for a set $S$, a bidder $b$, and prices $q$ the requirement function 
$f_{b,q}(S) = \min_{D \in D_b(q)} \{\lvert D \cap S \rvert\}$.
\end{definition}
\begin{observation}[compare Lemma~$2.10$ in~\cite{Ben-ZwiLaNe2013}]\label{obs:req}
For a set $S$, a bidder $b$, and prices $q$, 
% and $\delta_J$ the vector with value $\delta$ to each item in $J$ and $0$ otherwise,
we have $f_{b,q}(S) \ge (w_{b}(q) - w_{b}(q+\delta_S)) / \delta$.
\end{observation}
\begin{proof}
Let $D' = \argmin_{D \in D_b(q)} \{\lvert D \cap S \rvert\}$.
Then $ w_{b}(q+\delta_S) \ge u_{b,q+\delta_S}(D') = u_{b,q}(D') - \delta f_{b,q}(S)
= w_{b}(q) - \delta f_{b,q}(S)$,
that is, $\delta f_{b,q}(S) \ge w_{b}(q) - w_{b}(q+\delta_S)$.
\end{proof}

\begin{observation} [Observation $3.2$ in~\cite{Ben-ZwiLaNe2013}]\label{obs:overdemanded}
If for a price vector~$q$ there exists~$S$ such that $\sum_b{f_{b,q}(S)} > |S|$, 
then $q$ is not envy free. In this case we call $S$ \emph{over-demanded} at 
prices~$q$.
\end{observation}
\begin{proof}
In any envy-free allocation of $S$, bidder $b$ must receive a set from his demand, 
thus $b$ must receive at least $f_{b,q}(S)$ many items of $S$.
As each item of $S$ is allocated to at most one bidder, it follows that at least $\sum_b f_{b,q}(S) > |S|$ many items of $S$ are
allocated in any envy-free allocation. Contradiction.
\end{proof}

Note that the utilities $w_b(q)$ and prices $q$ are a feasible solution to the
dual of the \newlp{} for any prices with $q \ge r$. Further note that any
optimal solution to the dual of the \newlp{} implies that there exists a 
corresponding \newwe{}.
By optimality of $p^*$ and the assumption that no \newwe{} allocation exists
for the prices $p^{*-\delta J}$, the objective value of the dual for $p^{*-\delta J}$ is 
strictly greater than the objective of the dual for $p^*$, i.e., 
$\sum_b{w_b(p^{*-\delta J})} + p^{*-\delta J}(\Omega) > \sum_b{w_b(p^*)} + p^*(\Omega)$.
Now by definition of $p^{*-\delta J}$ we know that $p^{*-\delta J}(\Omega) + \delta|J| = p^*(\Omega)$, 
hence together we have that $\sum_b{w_b(p^{*-\delta J})} - \sum_b{w_b(p^*)} > \delta|J|$.
With $\sum_b{f_{b,p^{*-\delta J}}(J)} \ge \left(\sum_b{w_b(p^{*-\delta J})} - 
\sum_b{w_b(p^*)}\right) / \delta$ by Observation~\ref{obs:req} we have that
$\sum_b{f_{b,p^{*-\delta J}}(J)} > \lvert J \rvert$ and thus the set $J$ 
is over-demanded at $p^{*-\delta J}$ by Observation~\ref{obs:overdemanded}.

Next we use the following theorem by Gul and Stacchetti~\cite{GulSt2000} 
to show that this implies $\sum_b{f_{b,p}(J)} > |J|$, i.e., a contradiction 
to the assumption that the prices $p$ are envy-free. By another result of
Gul and Stacchetti~\cite{GS1999}, monotone valuations that are gross substitute
also satisfy the single improvement property.
\begin{theorem} [Theorem $2$ in~\cite{GulSt2000}]
Let $q^{(1)},q^{(2)}$ be two price vectors such that $q^{(1)}\leq q^{(2)}$ and $S$ a set with $\forall j\in S,\ q^{(1)}(j)=q^{(2)}(j)$.
Then for a bidder $b$ that fulfills the single improvement property the following apply
\begin{enumerate}
	\item $f_{b,q^{(1)}}(S) \leq f_{b,q^{(2)}}(S)$
	\item $f_{b,q^{(1)}}(\Omega\setminus S) \geq f_{b,q^{(2)}}(\Omega\setminus S)$
\end{enumerate}
\end{theorem}

Recall that by the definition of $J$ and $p^{*-\delta J}$,
we have $p_j \le p^{*-\delta J}_j$ for $j \in J$ and $p_j \ge p^{*-\delta J}_j$ for $j \not\in J$.
Now if we take $q^{(1)} = p^{*-\delta J}$ and take $q^{(2)}_j = p^{*-\delta J}_j \text{ if }j\in J$ and 
$q^{(2)}_j = p_j \text{ if }j\notin J$, then $q^{(1)} \le q^{(2)}$ and thus by the first 
part of the theorem we get that $f_{b,q^{(2)}}(J) \ge f_{b,p^{*-\delta J}}(J)> |J|$, i.e., the 
set $J$ is over-demanded at prices $q^{(2)}$. On the other hand, if we take 
$q^{(3)} = p$ and the same $q^{(2)}$, then $q^{(3)} \le q^{(2)}$ and $q^{(3)}_j = q^{(2)}_j$ for 
$j \not\in J$, and hence by the second part of the theorem with 
$S = \Omega \setminus J$ and thus $\Omega \setminus S = J$ we have that 
$f_{b,p}(J) \ge f_{b,q^{(2)}}(J)> |J|$. This shows that the set $J$ is over-demanded 
at $p$ as well and thus there cannot be an envy-free allocation for prices 
$p$ by Observation~\ref{obs:overdemanded}.
\end{proof}

\subsection{Computing an Equilibrium}\label{sec:algdesc}
The basic three-party auction is simple: First run the central
auction at the exchange, then the local auctions at the mediators.
In this section we summarize the details and analyze the time needed to 
compute a three-party competitive equilibrium.
We assume that all bidders have gross-substitute valuations and that
their valuations can be accessed via a demand oracle. We assume, for 
simplicity, that there are $m$ $\m$-mediators, 
each with $n/m$ distinct bidders.
We will use known polynomial-time auctions for the two-party allocation 
problem, see~\cite{PaesLeme2014} for a recent survey.
Theorem~\ref{th:reserve} shows how such an auction can be modified to yield
a \newwe{} instead of a Walrasian equilibrium.

Let~$A$ be a polynomial-time algorithm that can access 
$n$~gross-substitute valuations over subsets of $k$~items~$\Omega$
via a demand oracle
and outputs a Walrasian price vector $p\in \R^k$ and a Walrasian allocation 
$\langle \Omega_i \rangle_{i\in [n]}$.
Let the runtime of $A$ be  $T(n,k) = O(n^{\alpha}k^{\beta})$
for constants $\alpha$, $\beta$.

Although we can assume oracle access to the bidders' valuations,
we cannot assume it for the mediators' (gross-substitute) valuations, as they are not part of the
input. %\footnote{See Appendix~\ref{sec::oracle:lb}.}.
However, as outlined in the previous section, a mediator can determine
a set in her demand by running a single virtual auction to compute a \newwe{},
i.e., there is an efficient demand oracle for the mediators. 
Hence, solving the 
allocation problem for the central auction can be done in time
$T(m,k) \cdot T(n/m,k) = O(n^{\alpha}k^{2\beta})$.
Further, the local auctions for all mediators take time $O(m \cdot T(n/m,k))$
and thus the total time to compute a three-party competitive equilibrium is 
$O(n^{\alpha}k^{2\beta})$.
Note that the computation at the exchange takes only $T(m,k)$ time and
that the mediators are assumed to be separated, that is, the computation
at the mediators can be done in parallel.

% \section{The two-party allocation problem with few items in almost linear time}\label{sec:fewitems} %\label{sec:ud:reduction}
\subsubsection{Small Number of Items}\label{sec:fewitems}
In the context of ad exchanges it is 
natural to assume that the number of items is very small and independent of the 
number of bidders. In the following we discuss the computation of an equilibrium in 
this case. The results on this section will hold as long as the number
of items $k$ is $o(\log{n})$.

When the number of items is that small, bidders' valuations can be represented 
as complete lists. More than that, given a bidder valuation oracle, it takes 
only $2^k$ queries to compile such a list. In order to find the valuation 
lists of all the mediators as well, we have to solve the allocation
problem of each mediator $2^k$ times, i.e., compute the $\OR$ of the bidder 
valuations for all subsets. 
Given the valuations of the mediators, the central auction is equivalent
to solving the two-party allocation problem for the mediators. 
Let $T'(n, k)$ be the runtime of algorithm $A$ when valuations are accessed
via a valuation oracle. Then the overall running time to determine a 
three-party competitive equilibrium with this approach is 
$\widehat{T}(n,m,k) = m \cdot 2^k \cdot T'(n/m,k) + T'(m,k)$.

We show next how this approach can be extended to an almost linear time
algorithm for such a small number of items by artificially introducing 
mediators of mediators (and recurse).
Assume for simplicity $T'(n,k) = O(n^{\alpha}\cdot f(k))$ where $f(\cdot)$ is at most exponential in $k$ and $\alpha = 1 + \gamma$ for some $\gamma 
> 0$.\footnote{Current methods have $\alpha = 6$ and thus $\gamma = 5$.} By choosing $m = n^{1/2}$ 
we obtain a running time of
\begin{align*}
	\widehat{T}(n,n^{1/2},k)
	&= n^{1/2} \cdot 2^k \cdot 
n^{\alpha/2} \cdot f(k) + n^{\alpha/2} \cdot f(k)\,,\\
&=  \left(2^k n^{1+\gamma/2} + 
n^{1/2 + \gamma / 2}\right) \cdot f(k)\,,\\
&\le c \cdot n^{\alpha/2+1/2}\cdot 2^{2k}\,,
\end{align*}
for some constant $c \ge 0$.
Let us add one level of recursion:
\begin{align*}
	\widehat{T}(n,n^{1/2},k) &= 
	2^k \cdot n^{1/2} \cdot \widehat{T}(n^{1/2}, n^{1/4},k) \\&\phantom{=}+ 
	\widehat{T}(n^{1/2}, n^{1/4},k) \,,\\
	&\le c \cdot 2^{3k} \cdot n^{1/2} \cdot 
	(n^{1/2})^{\alpha/2 + 1/2} \\
	&\phantom{=}+ c \cdot 2^{2k} \cdot (n^{1/2})^{\alpha/2 + 1/2}\,,\\
	&\le c \cdot 2^{3k} \cdot ( n^{1/2 + 1/4 + \gamma/4 + 1/4} \\
	&\phantom{=}+ n^{1/4 + \gamma / 4 + 1 / 4} ) \,,\\
	&\le c \cdot 2^{3k} \cdot n^{(\alpha/2 + 1/2)/2 + 1/2} \,.
\end{align*}
For $t$ levels of mediators we obtain $\widehat{T}(n,n^{1/2},k) \le 
c \cdot 2^{(t+1)k} \cdot n^{\alpha_t/2 + 1/2}$ where $\alpha_0=\alpha$ and 
$\alpha_t = \frac{\alpha_{t-1}}{2}+\frac{1}{2} = \frac{\gamma}{2^t} + 1$.
Since for constant $\delta = 1/(t+1)$ we have that $k = o(\log{n})$ implies $ k = 
o(\log{n^{\delta}})$ and $\alpha$ is constant, we can choose $t$ to 
achieve a runtime of $O(n^{1+\epsilon+o(1)})$ for any fixed $\epsilon > 0$.

% Recall that a two-party Walrasian equilibrium can be obtained from a 
% three-party equilibrium. Thus the two-party allocation problem can be solved 
% in almost linear time when all $n$ bidders are gross substitute and the number 
% of items is $o(\log{n})$. 
An almost linear time algorithm to solve the 
two-party allocation problem when $k = o(\frac{\log{n}}{\log\log{n}})$
can be obtained by reducing the problem to unit-demand valuations in the following
way.
% \footnote{A direct almost linear time algorithm to solve the two-party allocation problem when $k = o(\frac{\log{n}}{\log\log{n}})$
% can be obtained by reducing the problem to unit-demand valuations. Details are given in Appendix~\ref{sec:ud:reduction}.}
Assume there are $n$ bidders and $k = o(\frac{\log{n}}{\log\log{n}})$ items 
$\Omega$.  The following method computes in almost linear time an allocation 
between bidders and items that maximizes social welfare (i.e., $\sum_b 
v_b(\Omega_b)$), which is 
equal to a Walrasian allocation if it exists. Consider all possible 
partitions of the $k$ items from which there are $O(k^k) = O(2^{k\log k})$ many. For a partition $P$
let the sets in the partition be the new items and let the value of the bidders
for a new item be their value for the set. Define a unit-demand valuation
function for each bidder based on these values. Then solve the allocation problem
for the new items and the unit-demand valuations. The resulting allocation
maximizes social welfare for the given partition. Over all possible partitions the 
one with maximum social welfare yields the desired solution. For unit-demand 
valuations the allocation problem is equivalent to the maximum weight bipartite 
matching problem that can be solved with the Hungarian method in 
time~$O(nk^2)$~\cite{Frank2005}. Thus the total time is $O(nk^{k + 2})$.
%Note that the prices computed for the unit-demand valuations might be to low.

% \section{Gross substitute oracle impossibility}\label{sec::oracle:lb}
% Consider some \emph{universal oracle} which can answer
% queries about every gross substitute valuation. To some extent, a polynomial representation
% of a valuation can serve as such, but not vise versa.
% For the gross substitute valuation class such a representation do not exists.
% Knuth showed that matroid rank functions over a ground set of $m$ elements, cannot be represented
% in less than $2^m/{poly(m)}$, where $poly(m)$ is some polynomial in $m$~\cite{knuth1974};
% matroid rank functions are known to be a subclass of gross substitute valuations~\cite{Fujishige2005}.
% 
% Following the discussion on equilibrium computation in Section~\ref{sec:algdesc}
% and in particular, the computation for small number of items in it, %Section~\ref{sec:fewitems}
% we conclude that any superlinear algorithm for computing (classic or three party)
% equilibrium can be altered to achieve an acceleration of almost a square root
% of the original. This method can then be invoked recursively to achieve
% a better and better running time.
% 
% For known problems, this acceleration is unrealistic so we can infer
% some \emph{conditional impossibility results}, such as:
% \begin{claim}
% There does not exist a universal oracle for gross substitute
% valuations unless \emph{bipartite weighted matching}~\cite{HopcroftKa1973} can be solved in
% an almost linear running time, where by `almost linear' we mean $o(n^{1+\epsilon}),\ \forall \epsilon > 0$.
% \end{claim}

\subsection{Relation to Minimum Walrasian Prices}\label{sec:minWEprices}

The following lemma shows that the prices the bidders have to pay at the computed
three-party equilibrium, and thus the utilities they achieve, 
are the same as in a Walrasian equilibrium between bidders and items with 
minimum prices.

\begin{lemma}\label{lem:minWE}
Let allocation $\langle \Omega_\beta \rangle$ and prices~$q$ be a Walrasian
equilibrium with minimum prices for gross-substitute bidders~$\B$ and items 
$\Omega$. Let each bidder be connected to exactly one of $m$ $\m$-mediators and let 
$\B_i$ denote the set of bidders connected to mediator~$\M_i$.
Let $r, p^1, p^2, \dotsc, p^m$ be the price vectors of a three-party equilibrium
for the mediators and bidders and let
$\langle \Omega_i \rangle$ be the equilibrium allocation of items to mediators and 
$\langle \Omega_b \rangle$ the equilibrium allocation of items to bidders.
Let $r$ and $\langle \Omega_i \rangle$ be a Walrasian equilibrium with minimum
prices for the mediators and let $p_j = \max_i(p^i_j)$ for all items $j$. Then $p = q$.
\end{lemma}
\begin{proof}
As $\langle \Omega_b \rangle$ and $p$ form a Walrasian equilibrium for bidders 
$\B$ and items $\Omega$, we have $p \ge q$. Recall $p \ge r$. Further $\langle 
\Omega_\beta \rangle$, 
$\langle \Omega'_i = \cup_{\beta \in \B_i} \Omega_\beta \rangle$, and $r' = 
p^{1\prime} = \dotsb = p^{m\prime} = q$ form a three-party equilibrium 
(compare Theorem~\ref{th:exist}) and
$\langle \Omega'_i \rangle$ and $q$ provide a Walrasian equilibrium
for the mediators. Thus by the minimality of $r$ we have $q \ge r$.
Assume by contradiction that there exists an item $j$ with $p_j > q_j$.
By Lemma~\ref{lem:efminef} and Corollary~\ref{cor:priceeq} each price 
vector~$p^i$ is the minimum envy-free price vector $\ge r$ for the 
bidders~$\B_i$. Thus at prices $q \not\ge p$ there is no allocation that is 
envy-free for all bidders~$\B$, a contradiction.
\end{proof}

\section{Discussion and Future Directions}\label{sec:discuss}
We proposed a new model for auctions at ad exchanges. Our model is more 
general than previous models in the sense that it takes the incentives of all 
three types of participants into account and that it allows to express 
preferences over multiple items.
Interestingly, at least when gross-substitute valuations are considered,
this generality does not come at the cost of tractability,
as shown by our polynomial-time algorithm. 
Note that this is the most general
result we could expect in light of the classical (two-sided) literature
on combinatorial auctions. 

We considered the special case of a small number of items for which
we showed that existing polynomial-time
algorithms for two-party equilibria can be sped-up by adding mediators.

In our model, with gross-substitute bidders, the revenue of the mediators
only comes from decreasing the central seller's profit (stated formally in 
Section~\ref{sec:minWEprices}). This explains the willingness of the bidders to use
mediators. The central seller experiences a decrease in its workload
with the introduction of mediators, which may partially describe its own
inclination for participating in the market.

Since our model tries to capture a single ``user-impression-of-a-web-page'' sold at an ad exchange, a natural
follow up work will try to model what happens over time.
This direction should take into account the change in
the environment, as the bidders and their valuations as well as the mediators and their connections to bidders can be different for each user impression.

\section*{Acknowledgements}
\noindent We wish to thank Noam Nisan for helpful discussions.
This work was funded by the Vienna 
Science and Technology Fund (WWTF) through project ICT10-002.
Additionally the research leading to these results has received funding from 
the European Research Council under the European Union's Seventh Framework Programme 
(FP/2007-2013) / ERC Grant Agreement no. 340506.

\printbibliography[heading=bibintoc]

\end{document}